\documentclass[runningheads]{llncs}

\usepackage{amssymb}
\usepackage{afterpage}
\usepackage{amsmath}
\usepackage{latexsym}
\usepackage{color}
\usepackage{ifthen}
\usepackage{multirow}
\usepackage{arrows}                  
\usepackage{rotating}
\usepackage[hidelinks]{hyperref}
\usepackage{subfigure}

\newif\ifextended\extendedtrue

\usepackage{tikz}
\usetikzlibrary{automata,fit}
\usetikzlibrary{decorations.pathmorphing}
\tikzstyle{every picture}=[thick]
\tikzstyle{every loop}=[->]
\tikzstyle{every scope}=[>=latex]
\tikzstyle{dot}=[circle,thick,minimum size=0.5mm,fill=black, inner sep=1pt]
\tikzstyle{every state}=[draw=black,line width=.5pt,fill=white,minimum size=10pt,initial text=]

\usetikzlibrary{shapes}
\newlength{\arclength}
\usepackage{dsfont}      
\usepackage{mathtools}   
\usepackage{mathpartir}  
\usepackage{relsize}     
\usepackage{url}
\usepackage{xspace}
\usepackage{algorithm2e}

\newcommand\distr{\ensuremath{\mathsf{Distr}}}

\newcommand{\parl}{\mathop{||}}

\include{arrowSupport}

 \newcommand{\mtrans}[1]{\stackrel{\smash{\lower1pt\hbox{\scriptsize $#1$}\,}}{\lower1pt\hbox{$\rightsquigarrow$}}}

\newcommand{\itrans}[1]{\xhookrightarrow{\smash{\lower1pt\hbox{\scriptsize $\,\smash{#1}\,$}\vphantom{X}}}}

\makeatletter 
\providecommand\dotsum{\mathpalette\@dotted\sum \vphantom{\sum}} 
\def\@dotted#1#2{\ooalign{\hfil$#1 \bullet $\hfil\cr\hfil$#1#2$\hfil\cr}} 
\makeatother 
\DeclareMathOperator*{\psum}{%
\mathchoice%
{\sum\!\!\!\!\!\!\!\!\scalebox{1.1}{$\bullet$}\ \ }%
{\sum\!\!\!\!\!\!{\raisebox{0.5pt}{\scalebox{0.825}{$\bullet$}}}\,\,\,}%
{}{}}

\newcommand{\action}[3]{#1 \!\psum_{#2} #3 :}

\let\ifTIKZ\iffalse
\let\ifTIKZ\iftrue

 \usetikzlibrary{decorations.markings}
 \usetikzlibrary{shapes.geometric}

 \pgfdeclarelayer{edgelayer}
 \pgfdeclarelayer{nodelayer}
 \pgfsetlayers{main,edgelayer,nodelayer}

\usetikzlibrary{arrows,shapes,petri,positioning,calc}
\ifTIKZ
\usetikzlibrary{decorations.text,decorations.pathmorphing}
\fi
\tikzstyle{place}=[circle,thick,draw=black!100,fill=black!0,minimum
size=3mm] 
\tikzstyle{stoTH}=[rectangle,thick, minimum width=1.5mm, minimum
height=5.5mm,thick,draw=black!100, fill=black!0,inner sep=0pt]
\tikzstyle{immTH}=[rectangle, minimum width=1mm, minimum
height=6mm,thin,draw=black!100, fill=black!100,inner sep=0pt]

\tikzstyle{aLabel}=[rectangle,draw=Red]

 \tikzstyle{ma-state}=[rectangle,fill=White,draw=Black]
 \tikzstyle{probability-state}=[circle,fill=Black,draw=Black,inner sep=0pt,minimum size=1pt]

\tikzstyle{arrow-segment-after-probabilities}=[->, >=stealth']
\tikzstyle{arrow-segment-before-probabilities-timed}=[draw=Black,postaction={decorate},
decoration={markings,mark=at position .5 with {\draw[fill=White] (-1pt,-3pt) rectangle (1pt,3pt) ;}}]
\tikzstyle{arrow-segment-before-probabilities-immediate}=[draw=Black,postaction={decorate},
decoration={markings,mark=at position .5 with {\draw[fill=Black] (-1pt,-3pt) rectangle (1pt,3pt) ;}}]

\tikzstyle{ma-trans-immedate}=[->, >=stealth',draw=Black,postaction={decorate},
decoration={markings,mark=at position .5 with {\draw[fill=Black] (-1pt,-3pt) rectangle (1pt,3pt) ;}}]

\tikzstyle{ma-trans-timed}=[->, >=stealth',draw=Black,postaction={decorate},
decoration={markings,mark=at position .5 with {\draw[fill=White] (-1pt,-3pt) rectangle (1pt,3pt) ;}}]

\tikzstyle{gspn-edge}=[->, >=stealth']

\newcommand{\paths}{\mathit{Paths}}
\newcommand{\nnreal}{\mathbb{R}_{\geq 0}}
\newcommand{\F}{\Diamond}
\newcommand{\Act}{{\sl Act}}
\renewcommand{\it}[1]{\singlearrow{#1}}  
\newcommand{\mt}[1]{\stackrel{#1}{\Longrightarrow}}  

\newcommand{\IS}{{\sl I\!S}}
\newcommand{\MS}{\mbox{\sl MS}}
\newcommand{\PS}{\mbox{\sl PS}}
\newcommand{\bfP}{\mathbf{P}}
\newcommand{\bfR}{\mathbf{R}}
\newcommand{\toolname}{{\sc MaMa}}  
\DeclareMathOperator*{\argmin}{arg\,min}
\newcommand{\G}{\Box}

\begin{document}

\title{\large Modelling, Reduction and Analysis of Markov Automata\ifextended\ (extended version)\fi\thanks{This work is
    funded by the EU FP7-projects MoVeS, SENSATION and MEALS, the
    DFG-NWO bilateral project ROCKS, the NWO projects SYRUP (grant
    612.063.817), the STW project ArRangeer (grant 12238), and the 
    DFG Sonderforschungsbereich AVACS.}} 

\author{Dennis Guck$^{1,3}$, Hassan Hatefi$^{2}$, Holger Hermanns$^{2}$, \\ Joost-Pieter Katoen$^{1,3}$ and Mark Timmer$^{3}$}
\institute{  
  {$^1$ Software Modelling and Verification, RWTH Aachen University, Germany} \\
  {$^2$ Dependable Systems and Software, Saarland University, Germany} \\
  {$^3$ Formal Methods and Tools, University of Twente, The Netherlands}
}

\titlerunning{Modelling, Reduction and Analysis of Markov Automata}
\authorrunning{Guck, Hatefi, Hermanns, Katoen, and Timmer}

\maketitle

\begin{abstract}
  Markov automata (MA) constitute an expressive continuous-time
  compositional modelling formalism. They appear as semantic
  backbones for engineering frameworks including dynamic fault trees,
  Generalised Stochastic Petri Nets, and AADL.  Their expressive power
  has thus far precluded them from effective analysis by probabilistic
  (and statistical) model checkers, stochastic game solvers, or
  analysis tools for Petri net-like formalisms.  This paper presents
  the foundations and underlying algorithms for efficient MA
  modelling, reduction using static analysis, and most importantly,
  quantitative analysis.  We also discuss implementation pragmatics of
  supporting tools and present several case studies demonstrating
  feasibility and usability of MA in practice.
\end{abstract}

\ifextended\else\vspace*{-.5cm}\fi

\section{Introduction} 

Markov automata (MA, for short) have been introduced in~\cite{EHZ10} as a continuous-time version of Segala's (simple) probabilistic automata~\cite{Seg95b}.
They are closed under parallel composition and hiding.
An MA-transition is either labelled with an action, or with a positive real number representing the rate of a negative exponential distribution.
An action transition leads to a discrete probability distribution over states.
MA can thus model action transitions as in labelled transition systems, probabilistic branching, as well as delays that are governed by exponential distributions.

The semantics of MA has been recently investigated in quite some detail.
Weak and strong (bi)simulation semantics have been presented in~\cite{EHZ10,EHZ10b}, whereas it is shown in~\cite{DBLP:journals/iandc/DengH13} that weak bisimulation provides a sound and complete proof methodology for reduction barbed congruence.
A process algebra with data for the efficient modelling of MA, accompanied with some reduction techniques using static analysis, has been presented in~\cite{MAPA}.
Although the MA model raises several challenging theoretical issues, both from a semantical and from an analysis point of view, our main interest is in their practical applicability.
As MA extend Hermanns' interactive Markov chains (IMCs)~\cite{Hermanns02}, they inherit IMC application domains, ranging from GALS hardware designs~\cite{CosteHLS09} and dynamic fault trees~\cite{DBLP:journals/tdsc/BoudaliCS10} to the standardised modeling language AADL~\cite{Bozzano,HaverkortKRRS10}. 
The added feature of probabilistic branching yields a natural operational model for generalised stochastic Petri nets (GSPNs)~\cite{MCB84} and stochastic activity networks (SANs)~\cite{MeyerMS85}, both popular modelling formalisms for performance and dependability analysis.
Let us briefly motivate this by considering GSPNs.
Whereas in SPNs all transitions are subject to a random delay, GSPNs also incorporate immediate transitions, transitions that happen instantaneously.
\begin{figure}[t!]
\centering
{}\hfill
\subfigure[]{
\scalebox{0.95}{

\begin{tikzpicture}
	\begin{pgfonlayer}{nodelayer}
		\node [style=place, tokens=1,label=$p_1$] (0) at (0, 0.75) {};
		\node [style=place,label=$p_3$] (1) at (1.5, 0.75) {};
		\node [style=place,label=$p_4$] (2) at (3.0, -0) {};
		\node [style=place,label=$p_5$] (3) at (3.0, -0.75) {};
		\node [style=place, tokens=1,label=below:$p_2$] (4) at (0, -0.75) {};
		\node [style=immTH,label=$t_1$] (5) at (0.75, 0.75) {};
		\node [style=immTH,label=$t_3 (w_3)$] (6) at (2.25, -0) {};
		\node [style=immTH,label=below:{$t_2 (w_2)$}] (7) at (0.75, -0.75) {};
		\node [style=stoTH,label=$\lambda_1$] (8) at (3.75, -0) {};
		\node [style=stoTH,label=below:{$\lambda_2$}] (9) at (3.75, -0.75) {};
		\node [style=place,label=$p_6$] (10) at (4.5, -0) {};
		\node [style=place,label=$p_7$] (11) at (4.5, -0.75) {};
	\end{pgfonlayer}
	\begin{pgfonlayer}{edgelayer}
		\draw [style=gspn-edge] (0) to (5);
		\draw [style=gspn-edge] (5) to (1);
		\draw [style=gspn-edge, bend right=15, looseness=1.00] (1) to (6);
		\draw [style=gspn-edge, in=180, out=90, looseness=0.50] (4) to (6);
		\draw [style=gspn-edge] (4) to (7);
		\draw [style=gspn-edge] (7) to (3);
		\draw [style=gspn-edge] (6) to (2);
		\draw [style=gspn-edge] (2) to (8);
		\draw [style=gspn-edge] (8) to (10);
		\draw [style=gspn-edge] (3) to (9);
		\draw [style=gspn-edge] (9) to (11);
	\end{pgfonlayer}
\end{tikzpicture}}
\label{fig:gspn-confused}
} \hfill
\subfigure[]{
\tikzstyle{every picture}=[thick, scale=0.78, transform shape]
\tikzstyle{every loop}=[->]
\tikzstyle{every scope}=[>=latex]
\centering
\begin{tikzpicture}[node distance=2.7cm,scale=0.9,every node/.style={transform shape}]
\tikzstyle{every label}=[label distance=0pt]
	\node[initial,state, ellipse, minimum width=43pt,fill=gray!40!white] (p12) {$p_1,p_2$};
	\node[state, ellipse, minimum width=43pt,fill=gray!40!white,above right of=p12,xshift=-0.5cm, yshift=-.5cm] (p23) {$p_2,p_3$};
	\node[state, ellipse, minimum width=43pt,fill=gray!40!white,below right of=p12,xshift=-0.5cm, yshift=.5cm] (p15) {$p_1,p_5$};
	\node[state, ellipse, minimum width=43pt,fill=gray!40!white,right of=p23] (p4) {$p_4$};
	\node[state, ellipse, minimum width=43pt,fill=gray!40!white,right of=p4] (p6) {$p_6$};
	\node[state, ellipse, minimum width=43pt,fill=gray!40!white,right of=p15] (p35) {$p_3,p_5$};
	\node[state, ellipse, minimum width=43pt,fill=gray!40!white,right of=p35] (p37) {$p_3,p_7$};
	
	\path[->] 
		(p4)  edge[thin,double] node[auto] {$\lambda_1$} (p6)
		(p35) edge[thin,double] node[auto] {$\lambda_2$} (p37);
		
	\path[->]
		(p12) edge[thin,bend left=20]  node[auto] {$\tau$} (p23)
		(p12) edge[thin,bend right=20] node[auto,swap] {$\tau$} (p15);
		
	\path[->]
		(p15) edge[thin] node[auto] {$\tau$} (p35);
		
	\path[->]
		(p23) edge[thin] node[auto] {$\dfrac{w_3}{w_2+w_3}$} (p4)
		(p23) edge[thin] node[auto,swap,yshift=0.3cm] {$\dfrac{w_2}{w_2+w_3}$} (p35);
		
	\path[-]
		(p23) edge[thin] node[inner sep=0mm,pos=0.2] (a1) {} (p4)
		(p23) edge[thin] node[inner sep=0mm,pos=0.2] (b1) {} (p35);	
	\path[-,shorten <=-.4pt,shorten >=-.4pt] (a1) edge [thin,bend left]  (b1) node[right,yshift=-.6cm] {$\tau$} ;

\end{tikzpicture}
\hfill{}
\label{fig:gspn-ma-semantics}
}
\caption{(a) Confused GSPN, see \cite[Fig.\ 21]{Mar95} with partial weights and (b) its MA semantics}
\end{figure}
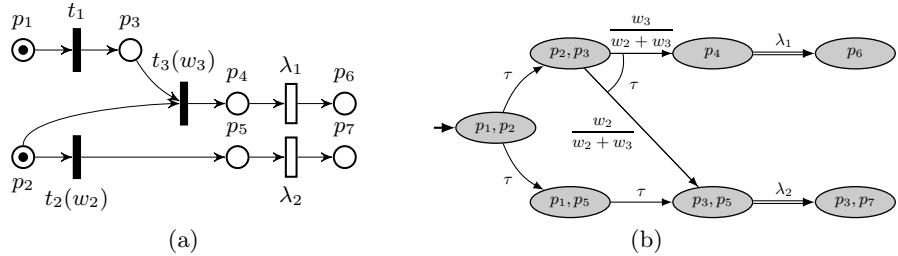
The traditional GSPN semantics yields a continuous-time Markov chain (CTMC), i.e., an MA without action transitions, but is restricted to GSPNs that do not exhibit non-determinism.
Such ``well-defined'' GSPNs occur if the net is free of confusion.
It has recently been detailed in~\cite{Katoen12,EHKZ13} that MA are a natural semantic model for \emph{every} GSPN.
Without going into the technical details, consider the confused GSPN in Fig.~\ref{fig:gspn-confused}.
This net is confused, as the transitions $t_1$ and $t_2$ are not in conflict, but firing transition $t_1$ leads to a conflict between $t_2$ and $t_3$, which does not occur if $t_2$ fires before $t_1$.
Transitions $t_2$ and $t_3$ are weighted so that in a marking $\{ p_2, p_3 \}$ in which both transitions are enabled, $t_2$ fires with probability $\frac{w_2}{w_2{+}w_3}$ and $t_3$ with its complement probability.
Classical GSPN semantics and analysis algorithms cannot cope with this net due to the presence of confusion (i.e., non-determinism). 
Figure~\ref{fig:gspn-ma-semantics} depicts the MA semantics of this net.
Here, states correspond to sets of net places that contain a token.
In the initial state, there is a non-deterministic choice between the transitions $t_1$ and $t_2$.
Note that the presence of weights is naturally represented by discrete probabilistic branching.
One can show that for confusion-free GSPNs, the classical semantics and the MA semantics are weakly bisimilar~\cite{EHKZ13}.

This paper focuses on the quantitative analysis of MA---and thus (possibly confused) GSPNs and probabilistic AADL error models.
We present analysis algorithms for three objectives: expected time, long-run average, and timed (interval) reachability.
As the model exhibits non-determinism, we focus on maximal and minimal values for all three objectives.
We show that expected time and long-run average objectives can be efficiently reduced to well-known problems on MDPs such as  stochastic shortest path, maximal end-component decomposition, and long-run ratio objectives. 
This generalizes (and slightly improves) the results reported in~\cite{DBLP:conf/nfm/GuckHKN12} for IMCs to MA.
Secondly, we present a discretisation algorithm for timed interval reachability objectives which extends~\cite{DBLP:conf/tacas/ZhangN10}.
Finally, we present the \toolname\ tool-chain, an easily accessible publicly available tool chain~\footnote{Stand-alone download as well as web-based interface available from \url{http://fmt.cs.utwente.nl/~timmer/mama}.} for the specification, mechanised simplification---such as confluence reduction~\cite{ConfluenceMA}, a form of on-the-fly partial-order reduction---and quantitative evaluation of MA. 
We describe the overall architectural design, as well as the tool components, and report on empirical results obtained with \toolname\ on a selection of case studies taken from different domains. 
The experiments give insight into the effectiveness of our reduction techniques and demonstrate that MA provide the basis of a very expressive stochastic timed modelling approach without sacrificing the ability of time and memory efficient numerical evaluation.

\vspace*{-.2cm}
\paragraph{Organisation of the paper.} 
After introducing Markov Automata in Section~\ref{sec:ma}, we discuss
a fully compositional modelling formalism in Section~\ref{section:mapa}. 
Section~\ref{section:expected} considers the evaluation of expected time 
properties. 
Section~\ref{section:longrun} discusses the analysis of long run properties, and
Section~\ref{section:timed} focusses on reachability properties with time interval 
bounds. 
Implementation details of our tool as well as experimental results are discussed 
in detail in Section~\ref{sec:tool}. 
Section~\ref{sec:conc} concludes the paper.
\ifextended Due to space constraints, we provide the proofs for our main results in appendices.
\else
Due to space constraints, we refer to~\cite{TechRep} for the proofs of our main results.
\fi
\ifextended\else\vspace*{-.3cm}\fi

\section{Preliminaries}
\label{sec:ma}
\vspace*{-.2cm}


%

\paragraph{Markov automata.}
An MA is a transition system with two types of transitions:  probabilistic (as in PAs) and Markovian transitions 
(as in CTMCs). 
Let $\Act$ be a universe  of actions with internal action $\tau \in \Act$, and $\distr(S)$ denote the set of distribution functions over the countable set $S$.

\begin{definition}[Markov automaton]\label{def:mas}
A \emph{Markov automaton (MA)} is a tuple $\mathcal{M} = \left( S, A,  \it{\, }, \mt{\, }, s_0 \right)$ where 
$S$ is a nonempty, finite set of states with \emph{initial state} $s_0 \in S$, $A \subseteq \Act$ is a finite set of actions, and 
\begin{itemize}
\item $\it{\, } \subseteq S \times A \times \distr(S)$ is the \emph{probabilistic} transition relation, and
\item $\mt{\, }\ \subseteq S \times \mathbb{R}_{> 0} \times S$ is the \emph{Markovian} transition relation. 
\end{itemize}
\end{definition}

We abbreviate $(s, \alpha, \mu) \in \it{\, }$ by $s \it{\alpha} \mu$ and $(s, \lambda, s') \in \ \mt{\, }$ by $\smash{s \mt{\lambda} s'}$.  
An MA can move between states via its probabilistic and Markovian transitions.
If~$s \it{a} \mu$, it can leave state $s$ by executing the action $a$, after which the probability to go to some state $s' \in S$ is given by $\mu(s')$. 
If $\smash{s \mt{\lambda} s'}$, it moves from~$s$ to~$s'$ with rate $\lambda$, except if $s$ enables a $\tau$-labelled transition. 
In that case, the MA will always take such a transition and never delays. 
This is the \emph{maximal progress} assumption~\cite{EHZ10}.
The rationale behind this assumption is that internal transitions are not subject to interaction and thus can happen immediately, whereas the probability for a Markovian transition to happen immediately is zero. As an example of an MA, consider Fig.~\ref{fig:MA}.

We briefly explain the semantics of Markovian transitions.
For a state with Markovian transitions, let $\bfR(s,s') = \smash{\sum \{ \lambda \mid s \mt{\lambda} s' \}}$ be the total rate to move from state~$s$ to state~$s'$, and let $E(s) = \sum_{s' \in S} \ \bfR(s,s')$ be the total outgoing rate of $s$.
If $E(s) > 0$, a competition between the transitions of $s$ exists.  
Then, the probability to move from $s$ to state $s'$ within $d$ time units is
$$
\frac{\bfR(s,s')}{E(s)} \cdot \left( 1 - e^{-E(s) d} \right).
$$
This asserts that after a delay of at most $d$ time units (second factor), the MA moves to a direct successor state $s'$ with probability $\bfP(s, s') = \frac{\bfR(s,s')}{E(s)}$. 

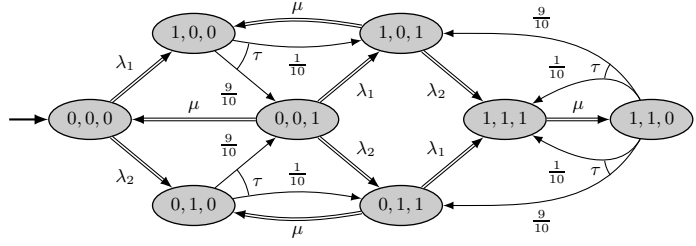
\begin{figure}[t]
\tikzstyle{every picture}=[thick, scale=0.78, transform shape]
\tikzstyle{every loop}=[->]
\tikzstyle{every scope}=[>=latex]
\centering\begin{tikzpicture}[node distance=2.5cm]
\tikzstyle{every label}=[label distance=0pt]

	\node[state, ellipse, minimum size=0pt,fill=gray!40!white] (s000)
{$0,0,0$};
	\node[] (sInit) [left of=s000, node distance=1.5cm] {};
	\draw[->] (sInit) edge [] node [auto, swap] {} (s000);
	\node[state, minimum size=0pt, draw=white] (s100) [above right of=s000]
{};
	\node[state, node distance=0.3cm, ellipse, minimum
size=0pt,fill=gray!40!white] (s100n) [below of=s100] {$1,0,0$};
	\node[state, minimum size=0pt, draw=white] (s010) [below right of=s000]
{};
	\node[state, node distance=0.3cm, ellipse, minimum
size=0pt,fill=gray!40!white] (s010n) [above of=s010] {$0,1,0$};
	\node[state, ellipse, minimum size=0pt,fill=gray!40!white] (s001) [below
right of=s100] {$0,0,1$};
	\node[state, minimum size=0pt, draw=white] (s101) [above right of=s001]
{};
	\node[state, node distance=0.3cm, ellipse, minimum
size=0pt,fill=gray!40!white] (s101n) [below of=s101] {$1,0,1$};
	\node[state, minimum size=0pt, draw=white] (s011) [below right of=s001]
{};
	\node[state, node distance=0.3cm, ellipse, minimum
size=0pt,fill=gray!40!white] (s011n) [above of=s011] {$0,1,1$};
	\node[state, ellipse, minimum size=0pt,fill=gray!40!white] (s111) [below
right of=s101] {$1,1,1$};
	\node[state, ellipse, minimum size=0pt,fill=gray!40!white] (s110) [right
of=s111] {$1,1,0$};

	\draw[->, double,thin] (s000) -- node [auto] {$\lambda_1$}
(s100n);
	\draw[->, double,thin] (s000) --  node [auto, swap]{$\lambda_2$}
(s010n);

	\draw[->,thin] (s100n) edge [] node [auto, swap] {$\frac{9}{10}$} (s001);
	\draw[->,thin] (s100n) edge [bend right=10] node [auto, swap] {$\frac{1}{10}$}
(s101n);
    \draw[thin] (s100n) + (1.25\arclength, -5pt) arc ( -10 :  -39.70265826695821
: 1.25\arclength);
	\path[thin] (s100n) +(31pt, -11.192068269362267pt) node {\phantom{.}} node
{$\tau$} ;

	\draw[->,thin] (s010n) edge [] node [auto] {$\frac{9}{10}$} (s001);
	\draw[->,thin] (s010n) edge [bend left=10] node [auto] {$\frac{1}{10}$}
(s011n);
    \draw[thin] (s010n) + (1.25\arclength, 5pt) arc ( 10 :  39.70265826695821 :
1.25\arclength);
	\path[thin] (s010n) +(31pt, 11.192068269362267pt) node {\phantom{.}} node
{$\tau$} ;

	\draw[->, double,thin] (s001) -- node [auto, swap] {$\mu$} (s000);
	\draw[->, double,thin] (s001) -- node [auto, swap] {$\lambda_1$}
(s101n);
	\draw[->, double,thin] (s001) -- node [auto] {$\lambda_2$}
(s011n);

	\draw[->, double,thin] (s101n) -- node [auto, swap] {$\lambda_2$}
(s111);
	\draw[->, double,thin] (s101n) [bend right=10] to node [auto,
swap] {$\mu$} (s100n);

	\draw[->, double,thin] (s011n) [bend left=10] to node [auto]
{$\mu$} (s010n);
	\draw[->, double,thin] (s011n) -- node [auto] {$\lambda_1$}
(s111);

	\draw[->, double,thin] (s111) -- node [auto] {$\mu$} (s110);

	\draw[->,thin] (s110) edge [in=0, out=120] node [pos=0.65, auto, swap]
{$\frac{9}{10}$} (s101n);
	\draw[->,thin] (s110) edge [in=30, out=120] node [pos=0.65, auto, swap]
{$\frac{1}{10}$} (s111);
    \draw[thin] (s110) + (-20pt,27pt) arc ( 145 :  172 : 0.75\arclength);
	\path[thin] (s110) +(-27pt, 25pt) node {\phantom{.}} node {$\tau$} ;

	\draw[->,thin] (s110) edge [in=0, out=-120] node [pos=0.65, auto]
{$\frac{9}{10}$} (s011n);
	\draw[->,thin] (s110) edge [in=-30, out=-120] node [pos=0.65, auto]
{$\frac{1}{10}$} (s111);
    \draw[thin] (s110) + (-20pt,-27pt) arc ( -145 :  -172 : 0.75\arclength);
	\path[thin] (s110) +(-27pt, -25pt) node {\phantom{.}} node {$\tau$} ;

\end{tikzpicture}
\caption{A queueing system, consisting of a server and two stations. The
two stations
have incoming requests with rates $\lambda_1, \lambda_2$, which are stored
until fetched by the server. If both stations contain a job, the server
chooses
nondeterministically (in state (1,1,0)). Jobs are processed with rate $\mu$, and when polling
a station, 
there is a $\frac{1}{10}$ probability that the job is erroneously kept in
the station 
after being fetched. Each state is represented as a tuple $(s_1,s_2,j)$,
with $s_i$ the number of jobs in station $i$, and $j$ the number of jobs
in the server.
For simplicity we assume that each component can hold at most one~job.}
\label{fig:MA}
\vspace*{-.3cm}
\end{figure}

\vspace*{-.2cm}
\paragraph{Paths.}
A path in an MA is an infinite sequence $ \pi \ = \ s_0 \it{\sigma_0,
  \mu_0, t_0} s_1 \it{\sigma_1, \mu_1, t_1} \ldots $ with $s_i \in S$,
$\sigma_i \in \Act \cup \{ \bot \}$, and $t_i \in \mathbb{R}_{\geq
  0}$.  For $\sigma_i \in \Act$, $s_i \it{\sigma_i, \mu_i, t_i}
s_{i+1}$ denotes that after residing $t_i$ time units in $s_i$, the MA
has moved via action $\sigma_i$ to $s_{i{+}1}$ with probability
$\mu_i(s_{i{+}1})$.  Instead, $s_i \it{\bot, \mu_i, t_i} s_{i+1}$
denotes that after residing $t_i$ time units in $s$, a Markovian
transition led to $s_{i+1}$ with probability $\mu_i(s_{i+1}) =
\bfP(s_i,s_{i+1})$.  For $t \in \mathbb{R}_{\geq 0}$, let $\pi@t$
denote the sequence of states that $\pi$ occupies at time $t$.  Due to
instantaneous action transitions, $\pi@t $ need not be a single state,
as an MA may occupy various states at the same time instant.  Let
$\paths$ denote the set of infinite paths.  The time elapsed along
the path $\pi$ is $\sum_{i=0}^{\infty}t_i$.  Path $\pi$ is Zeno
whenever this sum converges.  As the probability of a Zeno path in an
MA that only contains Markovian transitions is
zero~\cite{DBLP:journals/tse/BaierHHK03}, an MA is non-Zeno if and
only if no SCC with only probabilistic states is reachable with
positive probability. In the rest of this paper, we assume MAs to be
non-Zeno.

\vspace*{-.2cm}
\paragraph{Policies.}
Nondeterminism occurs when there is more than one action transition emanating from a state. 
To define a probability space, the choice is resolved using \emph{policies}.
A policy (ranged over by $D$) is a measurable function which yields for each finite path ending in state $s$ a probability distribution over the set of enabled actions in $s$. 
The information on basis of which a policy may decide yields different classes of policies.
Let $GM$ denote the class of the general measurable policies.
A stationary deterministic policy is a mapping $D \colon \PS \rightarrow \Act$ where $\PS$ is the set of states with outgoing probabilistic transitions; such policies always take the same decision in a state $s$.
A time-abstract policy may decide on basis of the states visited so far, but not on their timings; we use $TA$ denote this class.
For more details on different classes of policies (and their relation) on models such as MA, we refer to~\cite{NSK09}. Using a cylinder set construction we obtain a $\sigma$-algebra of subsets of $\paths$; given a policy~$D$ and an initial  state~$s$, a measurable set of paths is equipped with probability measure $\Pr_{s,D}$.

\vspace*{-.2cm}
\paragraph{Stochastic shortest path (SSP) problems.}
As some objectives on MA are reduced to SSP problems, we briefly introduce them.
A non-negative SSP problem is an MDP $(S,\Act,\bfP,s_0)$ with set $G \subseteq S$ of goal states, cost function
$c \colon S \setminus G \times \Act \to \mathbb{R}_{\geq 0}$ and terminal cost function $g\colon G \to \mathbb{R}_{\geq 0}$. 
The accumulated cost along a path~$\pi$ through the MDP before reaching $G$, denoted $C_G(\pi)$, is $\sum_{j{=}0}^{k{-}1} c(s_j,\alpha_j) + g(s_k)$ where $k$ is the state index of reaching $G$.
Let  $\mathit{cR}^{\min}(s, \F G)$ denote the minimum expected cost reachability of $G$ in the SSP when starting from $s$.
This expected cost can be obtained by solving an LP problem~\cite{BerTsi91}.

\ifextended\else\vspace*{-.3cm}\fi

\section{Efficient modeling of Markov automata}
\label{section:mapa}
\ifextended\else\vspace*{-.2cm}\fi


As argued in the introduction, MA can be used as semantical model for various modeling formalisms.
We show this for the process-algebraic specification language MAPA (MA Process Algebra)~\cite{MAPA}.
This language is rather expressive and supports several reductions techniques for MA specifications.
In fact, it turns out to be beneficial to map a language (like GSPNs) to MAPA so as to profit from these
reductions.
We present the syntax and a brief informal overview of the reduction techniques.  

\vspace*{-.2cm}
\paragraph{The Markov Automata Process Algebra.}

MAPA relies on external mechanisms for evaluating expressions, able to handle boolean and real-valued expressions. 
We assume that any variable-free expression in this language can be evaluated. 
Our tool uses a simple and intuitive fixed data language that includes basic arithmetic and boolean operators, conditionals, and dynamic lists. 
For expression $t$ in our data language and vectors $\vec{x} = (x_1, \dots, x_n)$ and $\vec{d} = (d_1, \dots, d_n)$, let $t[\vec{x}:=\vec{d}]$ denote the result of substituting every $x_i$ in $t$ by $d_i$.

A \emph{MAPA specification} consists of a set of uniquely-named \emph{processes}~$X_i$, each defined by a \emph{process equation} $X_i(\vec{x_i:D_i}) = p_i$. 
In such an equation, $\vec{x_i}$ is a vector of process variables with type~$\vec{D_i}$, and $p_i$ is a \emph{process term} specifying the behaviour of~$X_i$. 
Additionally, each specification has an \emph{initial process} $X_j(\vec{t})$. 
We abbreviate $X((x_1, \dots, x_n) : (D_1 \times \dots \times D_n))$ by $X(x_1 : D_1, \dots, x_n : D_n)$.
A MAPA \emph{process term} adheres to the grammar:
\vspace*{-0.1cm}
\[
    \textstyle p ::= Y(\vec{t}) \ \mid\  c \Rightarrow p \ \mid \ p + p \ \mid \ \sum_{\vec{x}:\vec{D}} p \ \mid \ \action{a(\vec{t})}{\vec{x} : \vec{D}}{f} p \ \mid \ (\lambda) \cdot p
\]
%
%
Here, $Y$ is a process name, $\vec{t}$ a vector of expressions, $c$ a boolean expression, $\vec{x}$~a vector of variables ranging over a finite type~$\vec{D}$, $a \in \textit{Act}$ a (parameterised) atomic action, $f$ a real-valued expression yielding  a value in $[0,1]$, and $\lambda$ an expression yielding a positive real number. 
Note that, if $|\vec{x}| > 1$,  $\vec{D}$ is a Cartesian product, as for instance in $\sum_{(m,i) : \{m_1,m_2\} \times \{1,2,3\}} \text{send}(m,i)\ldots$.
In a process term, $Y(\vec{t})$ denotes \emph{process instantiation}, where $\vec{t}$ instantiates $Y$'s process variables (allowing recursion).
The term $c \Rightarrow p$ behaves as $p$ if the \emph{condition} $c$ holds, and cannot do anything otherwise. The~$+$~operator denotes \emph{nondeterministic choice}, and~$\sum_{\vec{x} : \vec{D}} p$ a \emph{nondeterministic choice over data type}~$\vec{D}$. 
The term $\action{a(\vec{t})}{\vec{x} : \vec{D}}{f} p$ performs the action $a(\vec{t})$ and then does a \emph{probabilistic choice} over~$\vec{D}$. 
It uses the value \mbox{$f[\vec{x}:=\vec{d}]$} as the probability of choosing each~$\vec{d} \in \vec{D}$. 
We write $a(\vec{t}) \cdot p$ for the action $a(\vec{t})$ that goes to $p$ with probability $1$.
Finally, $(\lambda) \cdot p$ can behave as $p$ after a delay, determined by an exponential distribution with rate $\lambda$.
Using MAPA processes as basic building blocks, the language also supports the modular construction of large systems via top-level parallelism (denoted $\parl$), encapsulation (denoted $\partial$), hiding (denoted $\tau$), and renaming (denoted $\gamma$), cf.\ \cite[App.\ B]{MAtechrep}. 
The operational semantics of a MAPA specification yields an MA; for details we refer to~\cite{MAPA}.
\vspace*{-.05cm}


\begin{example}
\begin{figure}[t]
\centering\framebox{
\scalebox{0.9}{\begin{minipage}{0.97\textwidth}
\vspace*{-0.45cm}
\begin{align*}
& \textbf{constant} \textit{ queueSize} = 10, \textit{nrOfJobTypes} = 3\\ 
& \textbf{type} \textit{ Stations} = \{1, 2\}, \textit{ Jobs} = \{1, \dots, \textit{nrOfJobTypes}\}\\[5pt]
   &\textit{Station}(i:\textit{Stations}, q:\text{Queue}, \textit{size}:\{0..\textit{queueSize}\}) \\
  &\qquad{}= \textit{size} < \textit{queueSize} \Rightarrow (2i+1) \cdot \smash{\textstyle\sum_{j:\textit{Jobs}} \textit{arrive}(j) \cdot \textit{Station}(i,\text{enqueue}(q,j), size+1)}\\
  &\qquad{}+ \textit{size} > 0 \hskip32pt \Rightarrow \textit{deliver}(i,\text{head}(q)) \smash{\psum_{k \in \{1,9\}}} \tfrac{k}{10} : k=1 \Rightarrow \textit{Station}(i,q,\textit{size}) \\
  &\hphantom{\qquad+ \text{size}(q) > 0 \Rightarrow \textit{deliver}(i,\text{head}(q)) \psum_{k \in \{1,9\}} \tfrac{k}{10}}\hskip20pt + k=9 \Rightarrow \textit{Station}(i,\text{tail}(q),\textit{size}-1) \\
\\[-11pt]
 & \textit{Server} = \textstyle\sum_{n:\textit{Stations}} \sum_{j:\textit{Jobs}} \textit{poll}(n,j) \cdot (2*j) \cdot \textit{finish}(j) \cdot \textit{Server} \\
\\[-10pt]
&\gamma(\textit{poll}, \textit{deliver}) = \textit{copy} \quad \quad \quad  \quad \mbox{\small $//$ actions \textit{poll} and \textit{deliver} synchronise and yield action \textit{copy}} \\
\\[-13pt]
&\textit{System} = \tau_{\{\textit{copy}, \textit{arrive}, \textit{finish}\}} (\partial_{\{\textit{poll}, \textit{deliver}\}}(\textit{Station}(1,\text{empty},0) \parl \textit{Station}(2,\text{empty},0) \parl \textit{Server}))
\end{align*}
\vspace*{-0.55cm}
\end{minipage}}
}
\vspace*{-0.25cm}
\caption{MAPA specification of a polling system.}
\label{fig:spec}
\vspace*{-.3cm}
\end{figure}

Fig.~\ref{fig:spec} depicts the MAPA specification~\cite{MAPA} of a polling system---inspired by~\cite{Polling}---which generalised the system of Fig.~\ref{fig:MA}. Now, there are incoming requests of 3 possible types, each of which has a different service rate. Additionally, the stations store these in a queue of size $10$.
\qed
\end{example}

\vspace*{-.3cm}
\paragraph{Reduction techniques.}
To simplify state space generation and reduction, we use a linearised format referred to as MLPPE (Markovian linear probabilistic process equation).
In this format, there is precisely one process consisting of a nondeterministic choice between a set of summands. 
Each summand can contain a nondeterministic choice, followed by a condition, and either an interactive action with a probabilistic choice (determining the next state) or a rate and a next state. 
Every MAPA specification can be translated efficiently into an MLPPE~\cite{MAPA} while preserving strong bisimulation.
On MLPPEs two types of reduction techniques have been defined: simplifications and state space reductions:

\begin{itemize}
\item 
\emph{Maximal progress reduction} removes Markovian transitions from states also having $\tau$-transitions.  It is more efficient to perform this on MLPPEs than on the initial MAPA specification.  We use heuristics (as in~\cite{PT09}) to omit all Markovian summands in presence of internal non-Markovian ones.
\item
\emph{Constant elimination}~\cite{KPST11} replaces MLPPE parameters that remain constants by their initial value.
\item 
\emph{Expression simplification}~\cite{KPST11} evaluates functions for which all parameters are constants and applies basic laws from logic. 
\item 
\emph{Summation elimination}~\cite{KPST11} removes unnecessary summations, transforming e.g.,
$\sum_{d:\mathbb{N}} d = 5 \Rightarrow \textit{send}(d) \cdot X$ to $\textit{send}(5) \cdot X$, $\sum_{d:\{1,2\}} a \cdot X$ to $a \cdot X$, and $\sum_{d:D} (\lambda) \cdot X$ to $(|D| \times \lambda) \cdot X$, to preserve the total rate to $X$.
\item
\emph{Dead-variable reduction}~\cite{PT09} detects states in which the value of some data variable $d$ is irrelevant. This is the case if $d$ will be overwritten before being used for all possible futures. Then, $d$ is reset to its initial value.
\item
\emph{Confluence reduction}~\cite{ConfluenceMA} detects spurious nondeterminism, resulting from parallel composition. It denotes a subset of the probabilistic transitions of a MAPA specification as confluent, meaning that they can safely be given priority if enabled together with other transitions. 
\end{itemize}

\ifextended\else\vspace*{-.3cm}\fi
\section{Expected time objectives}
\label{section:expected}

The actions of an MA are only used for composing models from smaller ones.
For the analysis of MA, they are not relevant and we may safely assume that all actions are internal\footnote{Like in the MAPA specification of the queueing system in Fig.~\ref{fig:spec}, the actions used in parallel composition are explicitly turned into internal actions by hiding.}.
Due to the maximal progress assumption, the outgoing transitions of a state $s$ are all either probabilistic transitions or Markovian transitions.
Such states are called probabilistic and Markovian, respectively; let $\PS \subseteq S$ and $\MS \subseteq S$ denote these sets.

Let $\mathcal{M}$ be an MA with state space $S$ and $G \subseteq S$ a set of goal states.
Define the (extended) random variable $V_G \colon \paths \rightarrow \nnreal^{\infty}$ as the elapsed time before first visiting some state in $G$. That is, for an infinite path $\pi = s_0 \smash{\xrightarrow{\sigma_0,\mu_0,t_0}} s_1 \smash{\xrightarrow{\sigma_1,\mu_1,t_1}} \cdots$, let $V_G(\pi) = \min \left\{ t \in \nnreal \mid G \cap \pi@t \not= \emptyset \right\}$ where $\min (\emptyset) = +\infty$.  
(With slight abuse of notation we use $\pi@t$ as the set of states occurring in the sequence $\pi@t$.)
The minimal expected time to reach $G$ from $s \in S$ is defined by
\ifextended\else\vspace*{-.2cm}\fi
\begin{align*}
  \mathit{eT}^{\min}(s, \F G) \ = \ 
  \inf_D \mathbb{E}_{s,D}(V_G) \ = \ 
  \inf_D \int_{\paths} \hspace{-2ex} V_G(\pi) \; \Pr\nolimits_{s,D}(d\pi)
\end{align*}
where $D$ is a policy on $\mathcal{M}$.
Note that by definition of $V_G$, only the amount of time before entering the first $G$-state is relevant.
Hence, we may turn all $G$-states into absorbing Markovian states without affecting the expected time reachability. 
In the remainder we assume all goal states to be absorbing.

\begin{theorem}\label{thm_expected_reachability}
  The function $\mathit{eT}^{\min}$ is a fixpoint of the Bellman operator
  {\small \begin{align*}
    \left[L(v)\right](s) = \begin{cases}
      \displaystyle \frac{1}{E(s)} + \sum_{s' \in S} \bfP(s,s') \cdot v(s') & \text{ if } s \in \MS \setminus G \\
      \displaystyle \min_{\alpha \in \textit{\footnotesize Act}(s)} \sum_{s' \in S} \mu^s_\alpha(s') \cdot v(s') & \text{ if } s \in \PS \setminus G \\
      \displaystyle 0 & \text{ if } s \in G.
    \end{cases}
  \end{align*}}%
\end{theorem}
For a goal state, the expected time obviously is zero.
For a Markovian state $s \not\in G$, the minimal expected time to $G$ is the expected sojourn time in $s$ plus the expected time to reach $G$ via its successor states.
For a probabilistic state, an action is selected that minimises the expected reachability time according to the distribution $\mu^s_\alpha$ corresponding to $\alpha$.
The characterization of $\mathit{eT}^{\min}(s, \F G)$ in Thm.\,\ref{thm_expected_reachability} allows us to
reduce the problem of computing the minimum expected time reachability in an MA to a non-negative 
SSP problem~\cite{BerTsi91,deAlf99}.

\begin{definition}[SSP for minimum expected time reachability]\label{def_expected_time_reachability_ssp}
The SSP of MA $\mathcal{M} = \left( S, \Act, \it{\, }, \mt{\, }, s_0 \right)$ for the expected time reachability of $G \subseteq S$ is $\mbox{\sf ssp}_{et}(\mathcal{M}) = \left( S, \Act \cup \left\{ \bot \right\}, \bfP, s_0, G,c, g \right)$ where $g(s) = 0$ for all $s \in G$ and
  {\small \begin{align*}
    \bfP(s, \sigma, s') & = \begin{cases}
      \frac{\bfR(s,s')}{E(s)} & \text{if } s \in \MS, \sigma = \bot \\
      \mu^s_\sigma(s') & \text{if } s \in \PS, s \it{\sigma} \mu^s_\sigma \\
      0 & \text{otherwise, and}
    \end{cases} &
    c(s,\sigma) & = \begin{cases}
      \frac{1}{E(s)} & \text{if } s \in \MS \setminus G, \sigma = \bot \\
      0 & \text{otherwise.}
    \end{cases}
  \end{align*}}
\end{definition}
\ifextended\else\vspace*{-.1cm}\fi
Terminal costs are zero.
Transition probabilities are defined in the standard way.
The reward of a Markovian state is its expected sojourn time, and zero otherwise.
\ifextended\else\vspace*{-.4cm}\fi
\begin{theorem}\label{thm_expected_time_reachability_reduction}
For MA $\mathcal{M}$, $\mathit{eT}^{\min}(s, \F G)$ equals $\mathit{cR}^{\min}(s, \F G)$ in 
$\mbox{\sf ssp}_{et}(\mathcal{M})$.
\end{theorem}
Thus here is a stationary deterministic policy on $\mathcal{M}$ yielding $\mathit{eT}^{\min}(s, \F G)$.
Moreover, the uniqueness of the minimum expected cost of an SSP~\cite{BerTsi91,deAlf99} now yields that
$\mathit{eT}^{\min}(s, \F G)$ is the unique fixpoint of $L$ (see Thm.\,\ref{thm_expected_reachability}).
The uniqueness result enables the usage of standard solution techniques such as value iteration and linear programming to compute $\mathit{eT}^{\min}(s, \F G)$.
For maximal expected time objectives, a similar fixpoint theorem is obtained, and it can be proven that those objectives correspond to the maximal expected reward in the SSP problem defined above.
In the above, we have assumed MA to not contain any Zeno cycle, i.e., a cycle solely consisting of probabilistic transitions.
The above notions can all be extended to deal with such Zeno cycles, by, e.g., setting the minimal expected time of states in Zeno BSCCs that do not contain $G$-states to be infinite (as such states cannot reach $G$).  
Similarly, the maximal expected time of states in Zeno end components (that do not containg $G$-states) can be defined as $\infty$, as in the worst case these states will never reach $G$.

\ifextended\else\vspace*{-.3cm}\fi

\newcommand{\mI}{\mathcal{M}}
\newcommand{\bdI}{\mathbf{1}}
\newcommand{\LRA}{\textit{LRA}}
\newcommand{\mR}{\mathcal{R}}
\newcommand{\Paths}{\paths}

\section{Long run objectives}
\label{section:longrun}
\ifextended\else\vspace*{-.2cm}\fi

Let $\mI$ be an MA with state space $S$ and $G \subseteq S$ a set of goal states.
Let $\bdI_G$ be the characteristic function of $G$, i.e., $\bdI_G(s) = 1$ if and only if $s \in G$.
Following the ideas of \cite{DBLP:conf/lics/Alfaro98,LHK01}, the fraction of time spent in $G$ on an infinite path $\pi$ in $\mI$ up to time bound $t \in \mathbb{R}_{\geq 0}$ is given by the random variable (r.\,v.)
$A_{G,t}(\pi) \ = \ \frac{1}{t} \int_0^t \bdI_G(\pi@u) \ du$.
Taking the limit $t \rightarrow \infty$, we obtain the r.\,v.\ 
\ifextended\else\vspace*{-.2cm}\fi
\begin{align*}
A_{G}(\pi) \ = \ \lim_{t\to\infty} A_{G,t}(\pi) \ = \ \lim_{t\to\infty} \frac{1}{t}\int_0^t \bdI_G(\pi@u) \ du.
\end{align*}
The expectation of $A_{G}$ for policy $D$ and initial state $s$ yields the corresponding long-run average time spent in $G$:
\begin{align*}
\LRA^D(s, G) 
= 
\mathbb{E}_{s,D}(A_{G}) 
= 
\int_{\paths} \hspace{-2ex} A_{G}(\pi) \, {\Pr}_{s,D}(d\pi).
\end{align*}
The minimum long-run average time spent in $G$ starting from state $s$ is then:
\begin{align*}
\LRA^{\min}(s,G) \ = \ \inf_D\ \LRA^D(s, G) \ = \ \inf_D \mathbb{E}_{s,D}(A_{G}).
\end{align*}
For the long-run average analysis, we may assume w.l.o.g.\ that $G\subseteq \MS$, as the long-run average time spent in any probabilistic state is always 0. 
This claim follows directly from the fact that probabilistic states are instantaneous, i.e.\ their sojourn time is $0$ by definition.
Note that in contrast to the expected time analysis, $G$-states cannot be made absorbing in the long-run average analysis.
It turns out that stationary deterministic policies are sufficient for yielding minimal or maximal long-run average objectives.

In the remainder of this section, we discuss in detail how to compute the minimum long-run average fraction of time to be in $G$ in an MA $\mI$ with initial state $s_0$.
The general idea is the following three-step procedure:
\begin{enumerate}
\item Determine the maximal end components\footnote{A sub-MA of MA $\mathcal{M}$ is a pair $(S',K)$ where $S' \subseteq S$ and $K$ is a function that assigns to each $s\in S'$ a non-empty set of actions such that for all $\alpha \in K(s)$, $s \it{\alpha} \mu$ with $\mu(s') > 0$ or $\smash{s \mt{\lambda} s'}$ imply $s' \in S'$. 
An end component is a sub-MA whose underlying graph is strongly connected; it is maximal w.r.t.\ $K$ if it is not 
contained in any other end component $(S'',K)$.} $\{ \mI_1, \ldots, \mI_k \}$ of MA $\mI$.
\item Determine $\LRA^{\min}(G)$ in maximal end component $\mI_j$ for all $j \in \{ 1, \ldots, k \}$.  
\item Reduce the computation of $\LRA^{\min}(s_0,G)$ in MA $\mI$ to an SSP problem.
\end{enumerate}
The first phase can be performed by a graph-based algorithm~\cite{deA97_thesis,DBLP:conf/soda/ChatterjeeH11}, whereas the last two phases boil down to solving LP problems.

\paragraph{Unichain MA.}
We first show that for unichain MA, i.e., MA that under any stationary deterministic policy yield a strongly connected graph structure, computing $\LRA^{\min}(s, G)$ can be reduced to determining long-ratio objectives in MDPs.  
Let us first explain such objectives.
Let $M=(S,\Act,\bfP,s_0)$ be an MDP.
Assume w.l.o.g.\ that for each state $s$ in $M$ there exists $\alpha \in \Act$ such that $\bfP(s,\alpha,s') > 0$.
Let $c_1, c_2\colon S \times (\Act \cup \left\{ \bot \right\})  \to \mathbb{R}_{\geq 0}$ be cost functions.
The operational interpretation is that a cost $c_1(s,\alpha)$ is incurred when selecting action $\alpha$ in state $s$, and similar for $c_2$.
Our interest is the \emph{ratio} between $c_1$ and $c_2$ along a path.
The \emph{long-run ratio} $\mathcal{R}$ between the accumulated costs $c_1$ and $c_2$ along the infinite path $\pi = s_0 \it{\alpha_0} s_1 \it{\alpha_1} \ldots$ in the MDP $M$ is defined by\footnote{In our setting, $\mathcal{R}(\pi)$ is well-defined as the cost functions $c_1$ and $c_2$ are obtained from non-Zeno MA.  Thus for any infinite path $\pi$, $c_2(s_j,\alpha_j) > 0$ for some index $j$.}:  
$$
\mathcal{R}(\pi) \ = \ \displaystyle \lim_{n \to \infty} \dfrac{\sum_{i=0}^{n-1} c_1(s_i,\alpha_i)}{\sum_{j=0}^{n-1} c_2(s_j,\alpha_j)}.
$$ 
The minimum long-run ratio objective for state $s$ of MDP $M$ is defined by:
\begin{align*}
  R^{\min}(s) \ = \ 
  \inf_D \mathbb{E}_{s,D}(\mR) \ = \ 
  \inf_D \sum_{\pi \in \Paths} \mR(\pi) \cdot \text{Pr}_{s,D}(\pi).
\end{align*}

Here, $\Paths$ is the set of paths in the MDP, $D$ an MDP-policy, and $\Pr$ the probability mass on MDP-paths.
From~\cite{deA97_thesis}, it follows that $R^{\min}(s)$ can be obtained by solving the following LP problem with real variables $k$ and $x_s$ for each $s \in S$: Maximize $k$ subject to:
$$
x_s \, \leq \, c_1(s,\alpha) - k \cdot c_2(s,\alpha) + \sum_{s' \in S} \bfP(s,\alpha,s') \cdot x_{s'} \quad
\mbox{ for each } s \in S, \alpha \in \Act.
$$
We now transform an MA into an MDP with 2 cost functions as follows.
\begin{definition}[From MA to two-cost MDPs]\label{def:MAtomdp}
Let $\mI = \left( S, \Act,  \it{\, }, \mt{\, }, s_0 \right)$ be an MA and $G \subseteq S$ a set of goal states.
The MDP $\mbox{\sf mdp}(\mI) = (S, \Act\cup\{\bot\}, \bfP, s_0)$ with cost functions $c_1$ and $c_2$, where 
$\bfP$ is defined as in Def.~\ref{def_expected_time_reachability_ssp}, and
{\small \begin{align*}
  c_1(s,\sigma) & = \begin{cases}
    \frac{1}{E(s)} & \text{if } s \in \MS \cap G \wedge \sigma = \bot \\
    0 & \text{otherwise,}
  \end{cases}
  &
  c_2(s,\sigma) & = \begin{cases}
    \frac{1}{E(s)} & \text{if } s \in \MS \wedge \sigma = \bot \\
    0 & \text{otherwise.}
  \end{cases}
\end{align*}}
\end{definition}
Observe that cost function $c_2$ keeps track of the average residence time in state~$s$ whereas $c_1$ only does so for states in $G$.

\begin{theorem}
For unichain MA $\mI$, $LRA^{\min}(s,G)$ equals $R^{\min}(s)$ in $\mbox{\sf mdp}(\mI)$.
\end{theorem}

To summarise, computing the minimum long-run average fraction of time that is spent in some goal state in $G \subseteq S$ in an unichain MA $\mI$ equals the minimum long-run ratio objective in an MDP with two cost functions.
The latter can be obtained by solving an LP problem.
Observe that for any two states $s$, $s'$ in a unichain MA, $\LRA^{\min}(s,G)$ and $\LRA^{\min}(s',G)$ coincide.
We therefore omit the state and simply write $\LRA^{\min}(G)$ when considering unichain MA.

\paragraph{Arbitrary MA.}
Let $\mI$ be an MA with initial state $s_0$ and maximal end components $\{ \mI_1, \ldots, \mI_k \}$ for $k > 0$ where MA $\mI_j$ has state space $S_j$.
Note that each $\mI_j$ is a unichain MA.
Using this decomposition of $\mI$ into maximal end components, we obtain the following result:

\begin{theorem}\label{lem:LRA_SSP}\footnote{This theorem corrects a small flaw in the corresponding theorem for IMCs in~\cite{DBLP:conf/nfm/GuckHKN12}.} 
For MA $\mI = (S, \Act, \it{\,}, \mt{\, }, s_0)$ with MECs $\{ \mI_1, \ldots, \mI_k \}$ 
with state spaces $S_1, \dots, S_k \subseteq S$, and set of goal states $G \subseteq S$:
\begin{align*}
  \LRA^{\min}(s_0,G) & = \inf_{D} \sum_{j=1}^{k} \LRA^{\min}_j(G) \cdot {\Pr}^{D}(s_0 \models \diamondsuit \Box S_j),
\end{align*}
where ${\Pr}^{D}(s_0 \models \diamondsuit \Box S_j)$ is the probability to eventually reach and continuously stay in some state in $S_j$ from $s_0$ under policy $D$ and $\LRA^{\min}_j(G)$ is the LRA of $G \cap S_j$ in unichain MA $\mI_j$.  
\end{theorem}

Computing minimal LRA for arbitrary MA is now reducible to a non-negative SSP problem.
This proceeds as follows.
In MA $\mI$, we replace each maximal end component $\mathcal{M}_j$ by two fresh states $q_j$ and $u_j$.
Intuitively, $q_j$ represents $\mathcal{M}_j$ whereas $u_j$ represents a decision state.
State $u_j$ has a transition to $q_j$ and contains all probabilistic transitions leaving $S_j$.
Let $U$ denote the set of $u_j$ states and $Q$ the set of $q_j$ states.

\begin{definition}[SSP for long run average]
The SSP of MA $\mathcal{M}$ for the LRA in $G \subseteq S$ is
$\mbox{\sf ssp}_{lra}(\mathcal{M}) = 
\left( S \setminus \smash{\bigcup_{i=1}^k} S_i \cup U \cup Q, \Act \cup \{ \bot \}, \bfP', s_0, Q, c, g \right)$, where
$g(q_i) = \LRA^{\min}_i(G)$ for $q_i \in Q$ and $c(s,\sigma) = 0$ for all $s$ and $\sigma\in \Act\cup\{\bot\}$.
$\bfP'$ is defined as follows. Let $S' =  S \setminus \smash{\bigcup_{i=1}^k} S_i$.  $\bfP'$ equals $\bfP$ for all $s,s' \in S'$.
For the new states $u_j$:
\ifextended\else\vspace*{-.1cm}\fi
{\small
  \begin{align*}
  \bfP'(u_j, \tau, s') & =  \bfP(S_j, \tau, s') \quad \text{\!if } s' \in S' \setminus S_j 
  &
  \mbox{\!\!and\!\! } \quad \bfP'(u_i, \tau, u_j) & = \bfP(S_i, \tau, S_j) \quad \text{\!for } i \neq j.
\end{align*}}%
Finally, we have: $\bfP'(q_j,\bot,q_j) = 1 = \bfP'(u_j, \bot,q_j)$ and $\bfP'(s, \sigma, u_j) = \bfP(s, \sigma, S_j)$.
\end{definition}
Here, $\bfP(s,\alpha, S')$ is a shorthand for $\sum_{s' \in S'} \bfP(s,\alpha,s')$; similarly, $\bfP(S',\alpha,s') = \sum_{s \in S'} \bfP(s, \alpha, s')$.
The terminal costs of the new $q_i$-states are set to $\LRA^{\min}_i(G)$.
\begin{theorem} \label{thm:LRA_SSP}
For MA $\mathcal{M}$, $\LRA^{\min}(s,G)$ equals $cR^{\min}(s,\diamondsuit U)$ in SSP $\mbox{\sf ssp}_{lra}(\mathcal{M})$.
\end{theorem}
\ifextended\else\vspace*{-.2cm}\fi
\ifextended\else\vspace*{-.3cm}\fi
\section{Timed reachability objectives}
\label{section:timed}
\ifextended\else\vspace*{-.2cm}\fi

 \newcommand{\MaS}{\text{MS}}
\newcommand{\InS}{\ensuremath{S_I}\xspace}
\newcommand{\HiS}{\ensuremath{S_H}\xspace}

\newcommand{\DiBrPr}{\ensuremath{\textrm{\bf P}}} 

\newcommand{\Pref}{\ensuremath{\textsl{Pref}}\xspace} 
\newcommand{\MAM}{\ensuremath{\mathcal{M}}\xspace} 
\newcommand{\npaths}{\ensuremath{\textsl{Paths}^{n}}\xspace} 
\newcommand{\fpaths}{\ensuremath{\textsl{Paths}^*}\xspace} 
\newcommand{\ipaths}{\ensuremath{\textsl{Paths}^{\omega}}\xspace} 
\newcommand{\tafpaths}{\ensuremath{\textsl{Paths}^*_{\textsl{ta}}}\xspace} 
\newcommand{\taipaths}{\ensuremath{\textsl{Paths}^{\omega}_{\textsl{ta}}}\xspace} 
\newcommand{\lasts}{\ensuremath{\textsl{last}}} 

\newcommand{\sigalg}{\ensuremath{\mathfrak{F}}} 

\newcommand{\GMS}{\ensuremath{\textsl{GM}}\xspace} 
\newcommand{\TAS}{\ensuremath{\textsl{TA}}\xspace} 

\newcommand{\PrMe}{\ensuremath{\textsl{Pr}}\xspace} 
\newcommand{\dd}{\, \mathrm{d}}
\newcommand {\defma}{\ensuremath{\mathcal{M}=(S, \Act, \it{}, \mt{}, s_0)}\xspace}
\newcommand{\rchgls}[1]{\ensuremath{\diamondsuit^{#1}G}\xspace}


  This section presents an algorithm that approximates time-bounded reachability probabilities
  in MA.  We start with a fixed point characterisation, and then explain how 
  these probabilities can be approximated using digitisation.

\ifextended\else\vspace*{-.2cm}\fi
\paragraph{Fixed point characterisation.}
Our goal is to come up with a fixed point characterisation for the maximum (minimum) probability to reach a set of goal states in a time interval. Let $\mathcal{I}$ and $\mathcal{Q}$ be the set
  of all nonempty nonnegative real intervals with real and rational
  bounds, respectively.
For interval $I \in \mathcal{I}$ and $t \in \mathbb{R}_{\ge 0}$, let $I \ominus t = \left \{ x - t \mid x \in I \wedge x \ge t \right \}$. 
Given MA $\mathcal{M}$, $I \in \mathcal{I}$ and a set $G \subseteq S$ of goal states, the set of all paths that reach some goal states within interval $I$ is denoted by $\diamondsuit^{I} \, G$. 
Let $p^{\mathcal{M}}_{\max}(s,\diamondsuit^I \, G)$ be the maximum probability of reaching $G$ within interval $I$ if starting in state $s$ at time $0$. 
Here, the maximum is taken over all possible general measurable policies.
The next result provides a characterisation of $p_{\max}^{\mathcal{M}}(s,\diamondsuit^I \, G)$ as a fixed point.
\begin{lemma} 
\label{fpc:ma} 
Let $\mathcal{M}$ be an MA, $G \subseteq S$ and $I \in \mathcal{I}$ with $\inf I = a$ and $\sup I = b$. Then, $p^{\mathcal{M}}_{\max}(s,\diamondsuit^I \, G)$ is the least fixed point of the higher-order operator
    $\Omega \colon (S \times \mathcal{I} \rightarrowtail [0,1])
    \rightarrowtail (S \times \mathcal{I} \rightarrowtail [0,1])$,
    which for $s \in \MS$ is given by:
$$        \Omega(F)(s,I)=
        \begin{cases}
          \int_{0}^{b}E(s)e^{-E(s)t}\sum_{s' \in S}\DiBrPr(s,\bot,s')F(s',I \ominus t)\dd t & \!s \notin G \\
          e^{-E(s)a} + \int_{0}^{a}E(s)e^{-E(s)t}\sum_{s' \in
            S}\DiBrPr(s,\bot,s')F(s',I \ominus t)\dd t & \!s \in G
        \end{cases}
$$
and for $s \in \PS$ is defined by:
      \begin{equation*}
        \Omega(F)(s,I)=
        \begin{cases}
          1 & s \in G \wedge a = 0 \\
          \max_{\alpha \in Act_{\setminus \bot}(s)} \sum_{s' \in S}
          \DiBrPr(s, \alpha, s')F(s',I) & \mbox{otherwise.}
        \end{cases}
      \end{equation*} 
  \end{lemma}
This characterisation is a simple generalisation of that for IMCs~\cite{DBLP:conf/tacas/ZhangN10}, reflecting the fact that taking an action from an probabilistic state leads to a distribution over the states (rather than a single state).
The above characterisation yields an integral equation system which is in general not directly tractable \cite{DBLP:journals/tse/BaierHHK03}. 
To tackle this problem, we approximate the fixed point characterisation using digitisation, extending ideas developed in \cite{DBLP:conf/tacas/ZhangN10}. 
We split the time interval into equally-sized digitisation steps, assuming a digitisation constant $\delta$, small enough such that with high probability at most one Markovian transition firing occurs in any digitisation step. 
This allows us to construct a digitised MA (dMA), a variant of a semi-MDP, obtained by summarising the behaviour of the MA at equidistant time points.  
Paths in a dMA can be seen as time-abstract paths in the corresponding MA, implicitly still counting digitisation steps, and thus discrete time. 
Digitisation of MA \defma and digitisation constant $\delta$, proceeds by replacing $\mt{}$ by $\mt{}_{\delta} \, = \{ \, (s, \mu^{s}) \mid s \in \MS \, \}$, where
    \begin{equation*}
      \mu^s (s') = 
      \begin{cases}
        (1-e^{-E(s)\delta})\DiBrPr(s, \bot, s') & \mbox{if } s' \neq s\\
        (1-e^{-E(s)\delta})\DiBrPr(s, \bot, s') + e^{-E(s)\delta} & \mbox{otherwise.}
      \end{cases}
    \end{equation*} 
Using the above fixed point characterisation, it is now possible to relate reachability probabilities in an MA $\mathcal{M}$ to reachability probabilities in its dMA $\mathcal{M}_{\delta}$.
\ifextended\else\vspace*{-.2cm}\fi
  \begin{theorem}\label{thm:tbr}
    Given MA $\mathcal{M}=(S,\Act, \it{}, \mt{}, s_0)$, $G \subseteq S$, interval 
    $I=[0,b] \in \mathcal{Q}$ with $b \ge 0$ and $\lambda = \max_{s \in \scriptsize\MS}E(s)$. 
    Let $\delta > 0$ be such that $b=k_b\delta$ for some $k_b \in \mathbb{N}$. 
    Then, for all $s \in S$ it holds that
    \begin{equation*}
      p^{\mathcal{M}_{\delta}}_{\max}(s, \diamondsuit^{[0,k_b]} \, G) 
      \ \leq \ 
      p^{\mathcal{M}}_{\max}(s, \diamondsuit^{[0,b]} \, G) 
      \ \leq \ 
      p^{\mathcal{M}_{\delta}}_{\max}(s, \diamondsuit^{[0,k_b]} \, G) + 1 - e^{- \lambda b}\big(1+ \lambda \delta\big)^{k_b}.
    \end{equation*}
  \end{theorem}
This theorem can be extended to intervals with non-zero lower bounds; for the sake of brevity, the details are omitted here.
%
The remaining problem is to compute the maximum (or minimum) probability to reach $G$ in a dMA within a step bound $k \in \mathbb{N}$.
Let $\diamondsuit^{[0,k]} \, G$ be the set of infinite paths in a dMA that reach a $G$ state within $k$ steps, and
$p^{\mathcal{D}}_{\max}(s,\diamondsuit^{[0,k]} \, G)$ denote the maximum probability of this set.
Then we have
$p^{\mathcal{D}}_{\max}(s,\diamondsuit^{[0,k]} \, G) = \sup_{D \in \TAS} \Pr_{s,D}(\diamondsuit^{[0,k]} \, G)$. 
Our algorithm is now an adaptation (to dMA) of the well-known value iteration scheme for MDPs.

The algorithm proceeds by backward unfolding of the dMA in an
iterative manner, starting from the goal states.  Each iteration
intertwines the analysis of Markov states and of probabilistic states.
The key issue is that a path from probabilistic states to $G$ is split
into two parts: reaching Markov states from probabilistic states in zero
time and reaching goal states from Markov states in interval $[0,j]$,
where $j$ is the step count of the iteration. The former computation
can be reduced to an unbounded reachability problem in the MDP induced
by probabilistic states with rewards on Markov states.  For the latter,
the algorithm operates on the previously computed reachability
probabilities from all Markov states up to step count $j$.  We can
generalize this recipe from step-bounded reachability to step
interval-bounded reachability, details are described
in~\cite{HatefiH12}.

\ifextended\else\vspace*{-.3cm}\fi

\section{Tool-chain and case studies}
\label{sec:tool}
\ifextended\else\vspace*{-.2cm}\fi

This section describes the implementation of the algorithms discussed, together with the modelling features resulting in our \toolname~tool-chain. 
Furthermore, we present two case studies that provide empirical evidence of the strengths and weaknesses of the \toolname~tool chain. 

\subsection{\toolname~tool chain} 

\begin{figure}[b]
\ifextended\else\vspace*{-.2cm}\fi
\centering\begin{tikzpicture}[scale=0.78, transform shape]

	\node[state, rectangle, minimum width=50pt, minimum height=20pt] (s_0) {SCOOP};
	\node[state, rectangle, minimum width=50pt, minimum height=20pt] (s_2) [right of=s_0, node distance=3.9cm] {IMCA};
	\node[state, rectangle, draw=white, minimum width=50pt, minimum height=20pt] (s_3) [right of=s_2, node distance=2.75cm] {Results};
	\node[state, rectangle, draw=white] (s_4) [above of=s_0, node distance=1.2cm] {MAPA spec + Property};
	\draw[draw=white] (s_0) -- node [auto,swap] {Goal states} (s_2);
	\draw[->] (s_0) -- node [auto] {MA} (s_2);
	\draw[->, in=295, out=245, loop] (s_0) edge node [auto,swap] {reduce} (s_0);
	\draw[->] (s_2) -- node [auto] {} (s_3);
	\draw[->] (s_4) -- (s_0);

	\node[state, draw=gray, rectangle, minimum width=50pt, minimum height=20pt] (s_5) [left of=s_0, node distance=3.9cm] {\color{gray} GEMMA};
	\draw[draw=white] (s_5) -- node [auto,swap] {\color{gray} Property} (s_0);
	\draw[->, gray] (s_5) -- node [auto] {\color{gray} MAPA-spec} (s_0);
	\node[state, rectangle, draw=white] (s_6) [above of=s_5, node distance=1.2cm] {\color{gray} GSPN + Property};
	\draw[->, gray] (s_6) -- (s_5);

\end{tikzpicture}
\ifextended\else\vspace*{-.4cm}\fi
\caption{Analysing Markov Automata using the \toolname~tool chain.}
\label{fig:approach}
\end{figure}
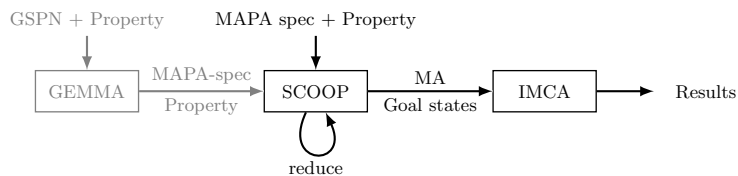

Our tool chain consists of several tool components: SCOOP~\cite{Timmer11,MAPA}, IMCA~\cite{DBLP:conf/nfm/GuckHKN12}, and GEMMA (realized in Haskell), see Figure~\ref{fig:approach}.
The tool-chain comprises about 8,000 LOC (without comments).
SCOOP (in Haskell) supports the generation from MA from MAPA specifications by a translation into the MLPPE format.
It implements all the reduction techniques described in Section~\ref{section:mapa}, in particular confluence reduction.
The capabilities of the IMCA tool-component (written in {\tt C++}) have been lifted to expected time and long-run objectives for MA, and extended with timed reachability objectives.
It also supports (untimed) reachability objectives which are not further treated here.
A prototypical translator from GSPNs to MA, in fact MAPA specifications, has been realized (the GEMMA component).
We connected the three components into a single tool chain, by making SCOOP export the (reduced) state space of an MLPPE in the IMCA input language. 
Additionally, SCOOP has been extended to translate properties, based on the actions and parameters of a MAPA specification, to a set of goal states in the underlying MA. 
That way, in one easy process systems and their properties can be modelled in MAPA, translated to an optimised MLPPE by SCOOP, exported to the IMCA tool and then analysed.

\ifextended\else\vspace*{-.2cm}\fi
\subsection{Case studies}
\ifextended\else\vspace*{-.2cm}\fi

This section reports on experiments with \toolname. All
experiments were conducted on a 2.5 GHz Intel Core i5 processor with 4GB RAM, running on Mac OS X 10.8.3. 

\ifextended\else\vspace*{-.2cm}\fi
\paragraph{Processor grid.} 
First, we consider a model of a $2 \times 2$ concurrent processor architecture. 
Using GEMMA, we automatically derived the MA model from the GSPN model in~\cite[Fig.~11.7]{Mar95}. 
Previous analysis of this model required weights for all immediate transitions, requiring complete knowledge of the mutual behaviour of all these transitions. 
We allow a weight assignment to just a (possibly empty) subset of the immediate transitions---reflecting the practical scenario of only knowing the mutual behaviour for a selection of the transitions. 
For this case study we indeed kept weights for only a few of the transitions, obtaining probabilistic behaviour for them and nondeterministic behaviour for the others.

Table~\ref{tab:grid_tb} reports on the time-bounded and time-interval bounded probabilities for reaching a state such that the first processor has an empty task queue. We vary the degree of multitasking $K$, the error bound $\epsilon$ and the interval~$I$. For each setting, we report the number of states $|S|$ and goal states $|G|$, and the 
generation time with SCOOP (both with and without the reductions from Section~\ref{section:mapa}).

The runtime demands grow with both the upper and lower time
bound, as well as with the required accuracy. The model size also
affects the per-iteration cost and thus the overall complexity of
reachability computation.
Note that our reductions speed-up the analysis times by a factor between $1.7$ and $3.5$: even more than the reduction in state space size. This is due to our techniques significantly reducing the degree of nondeterminism.

Table \ref{tab:grid} displays results for expected time until an empty task queue, as well as the long-run average that a processor is active. 
Whereas~\cite{Mar95} fixed all nondeterminism, obtaining for instance an LRA of $0.903$ for $K=2$, we are now able to retain nondeterminism and provide the more informative interval $[0.8810, 0.9953]$.
Again, our reduction techniques significantly improve runtimes.

\ifextended
  \begin{table}[p]
\else
  \begin{table}[t!]
\fi
\centering
{\smaller
\vskip5mm\hspace{-0.5cm}\scalebox{.93}{\begin{tabular}{c||ccc|ccc|cc|rrr|rrr}
& \multicolumn{3}{c|}{unreduced} & \multicolumn{3}{c|}{reduced}\\
$K$ & $|S|$ & $|G|$ & time & $|S|$ & $|G|$ & time & $\epsilon$ & $I$ &  \smash{\begin{turn}{35}$p^{\min}(s_0, \diamondsuit^{I} G)$ \end{turn}}\hspace*{-10ex}& \smash{\begin{turn}{35}time(unred)\end{turn}}\hspace*{-8ex} & \smash{\begin{turn}{35}time(red)\end{turn}}\hspace*{-6ex} &\smash{\begin{turn}{35} $p^{\max}(s_0, \diamondsuit^{I} G)$ \end{turn}}\hspace*{-10.5ex}&  \smash{\begin{turn}{35}time(unred)\end{turn}} \hspace*{-9ex} &  \smash{\begin{turn}{35}time(red)\end{turn}} \hspace*{-7ex}\\
\hline
\hline
\multirow{4}{*}{2} & \multirow{4}{*}{2{,}508}  & \multirow{4}{*}{1{,}398} & \multirow{4}{*}{0.6} & 
\multirow{4}{*}{1{,}789}  & \multirow{4}{*}{1{,}122} & \multirow{4}{*}{0.8} & 
$10^{-2}$ & $[0,3]$ & 0.91\phantom{0} & 58.5 & 31.0 & 0.95\phantom{0} & 54.9 & 21.7\\
  & & & & & & & 
$10^{-2}$ & $[0,4]$ & 0.96\phantom{0} & 103.0 & 54.7 & 0.98\phantom{0} & 97.3 & 38.8\\
  & & & & & & & 
$10^{-2}$ & $[1,4]$ & 0.91\phantom{0} & 117.3 & 64.4 & 0.96\phantom{0} & 109.9 & 49.0\\
  & & & & & & & 
$10^{-3}$ & $[0,3]$ & 0.910 & 580.1 & 309.4 & 0.950 & 544.3 & 218.4\\
\hline
\multirow{4}{*}{3} & \multirow{4}{*}{10{,}852}  & \multirow{4}{*}{4{,}504} & \multirow{4}{*}{3.1} & 
\multirow{4}{*}{7{,}201}  & \multirow{4}{*}{3{,}613} & \multirow{4}{*}{3.5} & 
$10^{-2}$ & $[0,3]$ & 0.18\phantom{6} & 361.5 & 202.8 & 0.23\phantom{1} & 382.8 & 161.1\\
  & & & & & & & 
$10^{-2}$ & $[0,4]$ & 0.23\phantom{6} & 643.1 & 360.0 & 0.30\phantom{1} & 681.4 & 286.0\\
  & & & & & & & 
$10^{-2}$ & $[1,4]$ & 0.18\phantom{6} & 666.6 & 377.3 & 0.25\phantom{1} & 696.4 & 317.7\\
  & & & & & & & 
$10^{-3}$ & $[0,3]$ & 0.176 & 3{,}619.5 & 2{,}032.1 & 0.231 & 3{,}837.3 & 1{,}611.9\\
\hline
4 & 31{,}832  & 10{,}424 & 9.8 & 20{,}021  & 8{,}357 & 10.5 & 
$10^{-2}$ & $[0,3]$ & 0.01\phantom{6} & 1{,}156.8 & 614.9 & 0.03\phantom{1} & 1{,}196.5 & 486.4\\
\hline
\end{tabular}}
}\vskip2pt
\caption{Interval reachability probabilities for the grid. (Time in seconds.)}
\label{tab:grid_tb}
\end{table}

\ifextended
  \begin{table}[p]
\else
  \begin{table}[t!]
\fi
\centering
{\smaller
\vskip5mm\hspace{-0.3cm}\scalebox{.93}{\begin{tabular}{c||rrr|rrr|rrr|rrr}
$K$ & \smash{\begin{turn}{35}$eT^{\min}(s_0, \F G)$ \end{turn}}\hspace*{-10ex}& \smash{\begin{turn}{35}time(unred)\end{turn}}\hspace*{-8ex} & \smash{\begin{turn}{35}time(red)\end{turn}}\hspace*{-6ex} &\smash{\begin{turn}{35} $eT^{\max}(s_0, \F G)$ \end{turn}}\hspace*{-10.5ex}&  \smash{\begin{turn}{35}time(unred)\end{turn}} \hspace*{-9ex} &  \smash{\begin{turn}{35}time(red)\end{turn}} \hspace*{-7ex} & 
\smash{\begin{turn}{35}$\LRA^{\min}(s_0, G)$ \end{turn}}\hspace*{-10ex}& \smash{\begin{turn}{35}time(unred)\end{turn}}\hspace*{-8ex} & \smash{\begin{turn}{35}time(red)\end{turn}}\hspace*{-6ex} &\smash{\begin{turn}{35} $\LRA^{\max}(s_0, G)$ \end{turn}}\hspace*{-10.5ex}&  \smash{\begin{turn}{35}time(unred)\end{turn}} \hspace*{-9ex} &  \smash{\begin{turn}{35}time(red)\end{turn}} \hspace*{-7ex}
\\
\hline
\hline
2 & 1.0000 & 0.3 & 0.1 & 1.2330 & 0.7 & 0.3 & 0.8110 & 1.3 & 0.7 & 0.9953 & 0.5 & 0.2\\
3 & 11.1168 & 18.3 & 7.7 & 15.2768 & 135.4 & 40.6 & 0.8173 & 36.1 & 16.1 & 0.9998 & 4.7 & 2.6\\ 
4 & 102.1921 & 527.1 & 209.9 & 287.8616 & 6{,}695.2 & 1{,}869.7 & 0.8181 & 505.1 & 222.3 & 1.0000 & 57.0 & 34.5\\ 
\hline
\end{tabular}}
}\vskip2pt
\caption{Expected times and long-run averages for the grid. (Time in seconds.)}
\label{tab:grid}
\ifextended\else\vspace*{-.7cm}\fi
\end{table}

\ifextended\else
\paragraph{Polling system.} 
Second, we consider the polling system from
Fig.~\ref{fig:spec} with two stations and one server.
We varied the queue sizes $Q$ and the number of job types~$N$,
analysing a total of six different settings.
Since---as for the previous case---analysis scales proportionally with the error bound, 
we \mbox{keep this constant here}.

Table~\ref{tab:poll_job_tb} reports results for
time-bounded and time-interval bounded properties, and
Table~\ref{tab:poll_job} displays probabilities and runtime
results for expected times and long-run averages. For all analyses,
the goal set consists of all states for which both station queues are full.
\fi

\ifextended
  \begin{table}[p]
\else
  \begin{table}[t!]
\fi
\centering
{\smaller
\vskip7mm\hspace{-0.4cm}\scalebox{.93}{\begin{tabular}{cc||ccc|ccc|cc|rrr|rrr}
&& \multicolumn{3}{c|}{unreduced} & \multicolumn{3}{c|}{reduced}\\
$Q$ & $N$ & $|S|$ & $|G|$ & time & $|S|$ & $|G|$ & time & $\epsilon$ & $I$ &  \smash{\begin{turn}{35}$p^{\min}(s_0, \diamondsuit^{I} G)$ \end{turn}}\hspace*{-10ex}& \smash{\begin{turn}{35}time(unred)\end{turn}}\hspace*{-8ex} & \smash{\begin{turn}{35}time(red)\end{turn}}\hspace*{-6ex} &\smash{\begin{turn}{35} $p^{\max}(s_0, \diamondsuit^{I} G)$ \end{turn}}\hspace*{-10.5ex}&  \smash{\begin{turn}{35}time(unred)\end{turn}} \hspace*{-9ex} &  \smash{\begin{turn}{35}time(red)\end{turn}} \hspace*{-7ex}\\
\hline
\hline
\multirow{2}{*}{2} & \multirow{2}{*}{3} & \multirow{2}{*}{1{,}497} & \multirow{2}{*}{567} & \multirow{2}{*}{0.4} & \multirow{2}{*}{990} & \multirow{2}{*}{324} & \multirow{2}{*}{0.2} &  
$10^{-3}$ & $[0,1]$ & 0.277 & 4.7 & 2.9 & 0.558 & 4.6 & 2.5\\
& & & & & & & &
$10^{-3}$ & $[1,2]$ & 0.486 & 22.1 & 14.9 & 0.917 & 22.7 & 12.5\\
\hline
\multirow{2}{*}{2} & \multirow{2}{*}{4} & \multirow{2}{*}{4{,}811} & \multirow{2}{*}{2{,}304} & \multirow{2}{*}{1.0} & \multirow{2}{*}{3{,}047} & \multirow{2}{*}{1{,}280} & \multirow{2}{*}{0.6} &  
$10^{-3}$ & $[0,1]$ & 0.201 & 25.1 & 14.4 & 0.558 & 24.0 & 13.5 \\
& & & & & & & &
$10^{-3}$ & $[1,2]$ & 0.344 & 106.1 & 65.8 & 0.917 & 102.5 & 60.5\\
\hline
\hline
\multirow{2}{*}{3} & \multirow{2}{*}{3} & \multirow{2}{*}{14{,}322} & \multirow{2}{*}{5{,}103} & \multirow{2}{*}{3.0} & \multirow{2}{*}{9{,}522} & \multirow{2}{*}{2{,}916} & \multirow{2}{*}{1.7} &  
$10^{-3}$ & $[0,1]$ & 0.090 & 66.2 & 40.4 & 0.291 & 60.0 & 38.5\\
& & & & & & & &
$10^{-3}$ & $[1,2]$ & 0.249 & 248.1 & 180.9 & 0.811 & 241.9 & 158.8\\
\hline
\multirow{2}{*}{3} & \multirow{2}{*}{4} & \multirow{2}{*}{79{,}307} & \multirow{2}{*}{36{,}864} & \multirow{2}{*}{51.6} & \multirow{2}{*}{50{,}407} & \multirow{2}{*}{20{,}480} & \multirow{2}{*}{19.1} &  
$10^{-3}$ & $[0,1]$ & 0.054 & 541.6 & 303.6 & 0.291 & 578.2 & 311.0\\
& & & & & & & &
$10^{-3}$ & $[1,2]$ & 0.141 & 2{,}289.3 & 1{,}305.0 & 0.811 & 2{,}201.5 & 1{,}225.9\\
\hline
\hline
\multirow{2}{*}{4} & \multirow{2}{*}{2} & \multirow{2}{*}{6{,}667} & \multirow{2}{*}{1{,}280} & \multirow{2}{*}{1.1} & \multirow{2}{*}{4{,}745} & \multirow{2}{*}{768} & \multirow{2}{*}{0.8} &  
$10^{-3}$ & $[0,1]$ & 0.049 & 19.6 & 14.0 & 0.118 & 19.7 & 12.8\\
& & & & & & & &
$10^{-3}$ & $[1,2]$ & 0.240 & 83.2 & 58.7 & 0.651 & 80.9 & 53.1\\
\hline
\multirow{2}{*}{4} & \multirow{2}{*}{3} & \multirow{2}{*}{131{,}529} & \multirow{2}{*}{45{,}927} & \multirow{2}{*}{85.2} & \multirow{2}{*}{87{,}606} & \multirow{2}{*}{26{,}244} & \multirow{2}{*}{30.8} &  
$10^{-3}$ & $[0,1]$ & 0.025 & 835.3 & 479.0 & 0.118 & 800.7 & 466.1\\
& & & & & & & &
$10^{-3}$ & $[1,2]$ & 0.114 & 3{,}535.5 & 2{,}062.3 & 0.651 & 3{,}358.9 & 2{,}099.5\\
\hline
\end{tabular}
}
}\vskip3pt
\caption{Interval reachability probabilities for the polling system. (Time in seconds.)}
\label{tab:poll_job_tb}
\ifextended\else\vspace*{-.5cm}\fi
\end{table}


\ifextended
  \begin{table}[p]
\else
  \begin{table}[t!]
\fi
\vskip0.7cm
\centering
{\smaller
\hspace{-0.3cm}\scalebox{.93}{\begin{tabular}{cc||rrr|rrr|rrr|rrr}
$Q$ & $N$ & \smash{\begin{turn}{35}$eT^{\min}(s_0, \F G)$ \end{turn}}\hspace*{-10ex}& \smash{\begin{turn}{35}time(unred)\end{turn}}\hspace*{-8ex} & \smash{\begin{turn}{35}time(red)\end{turn}}\hspace*{-6ex} &\smash{\begin{turn}{35} $eT^{\max}(s_0, \F G)$ \end{turn}}\hspace*{-10.5ex}&  \smash{\begin{turn}{35}time(unred)\end{turn}} \hspace*{-9ex} &  \smash{\begin{turn}{35}time(red)\end{turn}} \hspace*{-7ex} & 
\smash{\begin{turn}{35}$\LRA^{\min}(s_0, G)$ \end{turn}}\hspace*{-10ex}& \smash{\begin{turn}{35}time(unred)\end{turn}}\hspace*{-8ex} & \smash{\begin{turn}{35}time(red)\end{turn}}\hspace*{-6ex} &\smash{\begin{turn}{35} $\LRA^{\max}(s_0, G)$ \end{turn}}\hspace*{-10.5ex}&  \smash{\begin{turn}{35}time(unred)\end{turn}} \hspace*{-9ex} &  \smash{\begin{turn}{35}time(red)\end{turn}} \hspace*{-7ex}
\\
\hline
\hline
2 & 3 & 1.0478 & 0.2 & 0.1 & 2.2489 & 0.3 & 0.2 & 0.1230 & 0.8 & 0.5 & 0.6596 & 0.2 & 0.1\\
2 & 4 & 1.0478 & 0.2 & 0.1 & 3.2053 & 2.0 & 1.0 & 0.0635 & 9.0 & 5.2 & 0.6596 & 1.3 & 0.6\\
\hline
3 & 3 & 1.4425 & 1.0 & 0.6 & 4.6685 & 8.4 & 5.0 & 0.0689 & 177.9 & 123.6 & 0.6600 & 26.2 & 13.0 \\
3 & 4 & 1.4425 & 9.7 & 4.6 & 8.0294 & 117.4 & 67.2 & 0.0277 & 7{,}696.7 & 5{,}959.5 & 0.6600 & 1{,}537.2 & 862.4\\
\hline
4 & 2 & 1.8226 & 0.4 & 0.3 & 4.6032 & 2.4 & 1.6 & 0.1312 & 45.6 & 32.5 & 0.6601 & 5.6 & 3.9\\
4 & 3 & 1.8226 & 29.8 & 14.2 & 9.0300 & 232.8 & 130.8 & \multicolumn{3}{c|}{-- timeout (18 hours) --} & 0.6601 & 5{,}339.8 & 3{,}099.0\\
\hline
\end{tabular}}
}\vskip3pt
\caption{Expected times and long-run averages for the polling system. (Time in seconds.)}
\label{tab:poll_job}
\ifextended\else\vspace*{-.7cm}\fi
\end{table}
\ifextended
  \afterpage{\clearpage}
\fi

\ifextended\else\vspace*{-.2cm}\fi

\ifextended
\paragraph{Polling system.} 
Second, we consider the polling system from
Fig.~\ref{fig:spec} with two stations and one server.
We varied the queue sizes $Q$ and the number of job types~$N$,
analysing a total of six different settings.
Since---as for the previous case---analysis scales proportionally with the error bound, 
we \mbox{keep this constant here}.

Table~\ref{tab:poll_job_tb} reports results for
time-bounded and time-interval bounded properties, and
Table~\ref{tab:poll_job} displays probabilities and runtime
results for expected times and long-run averages. For all analyses,
the goal set consists of all states for which both station queues are full.
\fi

\ifextended\else\vspace*{-.3cm}\fi
\section{Conclusion}
\label{sec:conc}
\vspace*{-.2cm}

This paper presented new algorithms for the quantitative analysis of Markov automata (MA) and proved their correctness.
Three objectives have been considered: expected time, long-run average, and timed reachability.
The \toolname\ tool-chain supports the modelling and reduction of MA, and can analyse these three objectives.
It is also equipped with a prototypical tool to map GSPNs onto MA.
The \toolname\ is accessible via its easy-to-use web interface that can be found at \mbox{\url{http://wwwhome.cs.utwente.nl/~timmer/mama}}.
Experimental results on a processor grid and a polling system give insight into the accuracy and scalability of the presented algorithms.
Future work will focus on efficiency improvements and reward extensions.


\ifextended\else\vspace*{-.3cm}\fi
\ifextended\else\renewcommand{\baselinestretch}{0.85}\fi

\bibliographystyle{abbrv}
\bibliography{references}

\ifextended\else\renewcommand{\baselinestretch}{1.0}\fi

\ifextended
\appendix
\section{Proof of Theorem 1}
\setcounter{theorem}{0}
\begin{theorem}
  The function $\mathit{eT}^{\min}$ is a fixpoint of the Bellman operator
  {\small \begin{align*}
    \left[L(v)\right](s) = \begin{cases}
      \displaystyle \frac{1}{E(s)} + \sum_{s' \in S} \bfP(s,s') \cdot v(s') & \text{ if } s \in \MS \setminus G \\
      \displaystyle \min_{\alpha \in \textit{\footnotesize Act}(s)} \sum_{s' \in S} \mu^s_\alpha(s') \cdot v(s') & \text{ if } s \in \PS \setminus G \\
      \displaystyle 0 & \text{ if } s \in G.
    \end{cases}
  \end{align*}}
\end{theorem}
\textit{Proof}. We show that $L(eT^{\min}(s,\F G))=eT^{\min}(s,\F G)$, for all $s\in S$. Therefore, we will distinguish three cases: $s\in MS\setminus G, s\in PS\setminus G, s\in G$.
\renewcommand{\labelenumi}{(\roman{enumi})}
\begin{enumerate}
	\item if $s\in MS\setminus G$, we derive
	{\footnotesize
		\begin{eqnarray*}
			eT^{\min}(s,\F G) & = & \inf_D \mathbb{E}_{s,D}(V_G) = \inf_D \int_{\paths} V_G(\pi) \Pr_{s,D} (d\pi)\\
			& = & \inf_D \int_{0}^{\infty} t \cdot E(s)e^{-E(s)t} + \sum_{s \in S} \bfP(s,s') \cdot \mathbb{E}_{s',D(s\it{\bot, 1, t} \cdot)} (V_G)dt\\
			& = & \int_{0}^{\infty} t \cdot E(s)e^{-E(s)t} + \sum_{s \in S} \bfP(s,s') \cdot \inf_D \mathbb{E}_{s',D(s\it{\bot, 1, t} \cdot)} (V_G)dt\\
			& = & \int_{0}^{\infty} t \cdot E(s)e^{-E(s)t} + \sum_{s \in S} \bfP(s,s') \cdot \inf_D \mathbb{E}_{s',D} (V_G)dt\\
			& = & \int_{0}^{\infty} t \cdot E(s)e^{-E(s)t}dt + \sum_{s \in S} \bfP(s,s') \cdot eT^{\min}(s',\F G)\\
			& = & \frac{1}{E(s)} + \sum_{s \in S} \bfP(s,s') \cdot eT^{\min}(s',\F G)\\
			& = & L(eT^{\min}(s,\F G)).
		\end{eqnarray*}
	}
	\item if $s\in PS\setminus G$, we derive
		\begin{eqnarray*}
			eT^{\min}(s,\F G) & = & \inf_D \mathbb{E}_{s,D}(V_G) = \inf_D \int_{\paths} V_G(\pi)\Pr_{s,D}(d\pi)\\
			& = & \inf_D \sum_{s\it{\alpha,\mu,0}s'} D(s)(\alpha) \cdot \mathbb{E}_{s',D(s\it{\alpha,\mu,0}\cdot)}(V_G).
		\end{eqnarray*}
		Each action $\alpha\in\Act(s)$ uniquely determines a distribution $\mu^s_{\alpha}$, such that the successor state $s'$, with $s\it{\alpha,\mu^s_{\alpha},0}s'$, satisfies $\mu^s_{\alpha}(s') > 0$.
		\begin{equation*}
			\alpha = \argmin_{s\it{\alpha}\mu^s_{\alpha}}\inf_{D}\sum_{s'\in S}\mu^s_{\alpha}(s')\cdot \mathbb{E}_{s',D}(V_G)
		\end{equation*}
		Hence, all optimal schedulers choose $\alpha$ with probability $1$, i.e. $D(s)(\alpha) = 1$ and $D(s)(\sigma) = 0$ for all $\sigma\not=\alpha$. Thus, we obtain
		\begin{eqnarray*}
			eT^{\min}(s,\F G) & = & \inf_D \min_{s\it{\alpha}\mu^s_{\alpha}}\sum_{s'\in S}\mu^s_{\alpha}(s') \cdot \mathbb{E}_{s',D(s\it{\alpha,\mu^s_{\alpha},0}\cdot)}(V_G)\\
			& = & \min_{s\it{\alpha}\mu^s_{\alpha}}\inf_{D} \sum_{s' \in S}\mu^s_{\alpha}(s')\cdot \mathbb{E}_{s',D(s\it{\alpha,\mu^s_{\alpha},0}\cdot)}(V_G)\\
			& = & \min_{s\it{\alpha}\mu^s_{\alpha}}\inf_{D} \sum_{s' \in S}\mu^s_{\alpha}(s')\cdot \mathbb{E}_{s',D}(V_G)\\
			& = & \min_{s\it{\alpha}\mu^s_{\alpha}} \sum_{s' \in S}\mu^s_{\alpha}(s')\cdot eT^{\min}(s',\F G)\\
			& = & \min_{\alpha \in \Act(s)} \sum_{s' \in S}\mu^s_{\alpha}(s')\cdot eT^{\min}(s',\F G)\\
			& = & L(eT^{\min}(s,\F G)).
		\end{eqnarray*}
	\item if $s\in G$, we derive
		\begin{equation*}
			eT^{\min}(s,\F G) = \inf_{D}\int_{\paths}V_G(\pi)\Pr_{s,D}(d\pi) = 0 = L(eT^{\min}(s,\F G)).
		\end{equation*}
\end{enumerate} \qed
\section{Proof of Theorem 2}
\begin{theorem}
For MA $\mathcal{M}$, $\mathit{eT}^{\min}(s, \F G)$ equals $\mathit{cR}^{\min}(s, \F G)$ in 
$\mbox{\sf ssp}_{et}(\mI)$.
\end{theorem}
\textit{Proof}. As shown in \cite{BerTsi91,deA97_thesis}, $cR^{min}(s,\F G)$ is the unique fixpoint of the Bellman operator $L'$ defined as
\begin{equation*}
	[L'(v)](s)=\min_{\alpha\in\Act(s)} c(s,\alpha) + \sum_{s'\in S\setminus G} \bfP(s,\alpha,s')\cdot v(s') + \sum_{s'\in G}\bfP(s,\alpha,s')\cdot g(s').
\end{equation*}
We show that the Bellman operator $L$ for $\mI$ defined in Theorem 1 equals $L'$ for $\mbox{\sf ssp}_{et}(\mathcal{M})$. Note that by definition $g(s)=0$ for all $s\in G$. Thus
\begin{equation*}
[L'(v)](s)=\min_{\alpha\in\Act(s)} c(s,\alpha) + \sum_{s'\in S\setminus G} \bfP(s,\alpha,s')\cdot v(s').
\end{equation*}
We distinguish three cases, $s\in MS\setminus G, s\in PS\setminus G, s\in G$.
\renewcommand{\labelenumi}{(\roman{enumi})}
\begin{enumerate}
	\item If $s\in MS\setminus G$, then $|Act(s)|=1$ with $\Act(s)=\{\bot\}$ and therefore $\min_{\alpha \in \Act(s)} =\bot$. Further $c(s,\bot)=\frac{1}{E(s)}$ and for all $s'\in S,\bfP(s,\bot,s')=\frac{R(s,s')}{E(s)}$. Thus
	\begin{equation*}
		[L'(v)](s) = \frac{1}{E(s)} + \sum_{s'\in S} \frac{R(s,s')}{E(s)}\cdot v(s') = [L(v)](s).
	\end{equation*}
	\item If $s\in PS\setminus G$, for each action $\alpha\in\Act(s)$ and successor state $s'$, with $\bfP(s,\alpha,s') > 0$ it follows $\bfP(s,\alpha,s')=\mu^s_{\alpha}(s')$. Further, $c(s,\alpha) = 0$ for all $\alpha\in\Act$.
	{\small{
	\begin{equation*}
		[L'(v)](s) = \min_{\alpha\in\Act(s)} \sum_{s'\in S} \bfP(s,\alpha,s')\cdot v(s') = \min_{\alpha\in\Act(s)} \sum_{s'\in S} \mu^s_{\alpha}(s')\cdot v(s') = [L(v)](s).
	\end{equation*}}}
	\item If $s\in G$, then by definition $|Act(s)|=1$ with $\Act(s)=\{\bot\}$ and $\bfP(s,\bot,s)=1$ and $c(s,\bot)=0$.
	\begin{equation*}
		[L'(v)](s) = \sum_{s'\in S} \bfP(s,\alpha,s')\cdot v(s')  = 0 = [L(v)](s)
	\end{equation*}
\end{enumerate} \qed
\section{Proof of Theorem 3}
\begin{theorem}
For unichain MA $\mI$, $LRA^{\min}(s,G)$ equals $R^{\min}(s)$ in $\mbox{\sf mdp}(\mI)$.
\end{theorem}
\textit{Proof}. Let $\mI$ be an unchain MA with state space $S$ and $G\subseteq S$ a set of goal states. We consider a stationary deterministic scheduler $D$ on $\mI$. As $\mI$ is unIchain, $D$ will induce an ergodic CTMC with
\begin{equation*}
	\mathbf{R}(s,s')=\begin{cases}
		\sum\{\lambda|s\mt{\lambda}s'\} & \text{ if }s\in MS\\
		\infty & \text{ if }s\in PS \wedge s\it{D(s)}\mu^s_{D(s)}\wedge\mu^s_{D(s)}(s')>0 
	\end{cases}
\end{equation*}
Hence, the behaviour of Markovian states is the same as before. In contrary, for probabilistic states, the transitions induced by scheduler $D$ and probability distribution $\mu^s_{D(s)}$ are transformed into Markovian transitions with rate $\infty$. Thus, we simulate with the exponential distribution the instantaneous execution of the probabilistic transition. Note, that this will not contradict the applied results for CTMCs.\\
The long-run average for state $s\in S$ and a set of goal states $G$ is given by
 \begin{equation*}
	LRA^D(s,G)=\mathbb{E}_{s,D}(\mathcal{A}_G)=\mathbb{E}_{s,D}\left(\lim_{t\to\infty}\frac{1}{t}\int_0^t\textbf{1}_G(\mathcal{X}_u)du\right)
\end{equation*}
where $\mathcal{X}_u$ is the random variable, denoting the state $s$ at time point $u$. With the ergodic theorem from \cite{Nor97} we obtain the following:
 \begin{equation*}
	\mathbb{P}\left(\frac{1}{t}\int_0^t \textbf{1}_{\{x_s=i\}}ds \to \frac{1}{m_i q_i} \text{ as } t\to\infty\right)=1
\end{equation*}
where $m_i=E_i(T_i)$ is the expected return time to state $s_i$. Therefore, in our induced ergodic CTMC, almost surely
\begin{equation}
	\label{eq:fractionTime}
	\mathbb{E}_{s_i}\left(\lim_{t\to\infty}\frac{1}{t}\int_0^t\textbf{1}_{\{s_i\}}(\mathcal{X}_u)du\right) = \frac{1}{m_i\cdot E(s_i)}.
\end{equation}
Thus, the fraction of time to stay in $s_i$ in the long-run is almost surely $\frac{1}{m_i\cdot E(s_i)}$, where we assume that $\frac{1}{\infty}=0$.\\
Let $\mu_i$ be the probability to stay in $s_i$ in the long-run in the embedded DTMC of our ergodic CTMC, where $\bfP(s,s')=\frac{\bfR(s,s')}{E(s)}$. Thus $\mu\cdot\bfP=\mu$ where $\mu$ is the vector containing $\mu_i$ for all states $s_i\in S$. Given the probability of $\mu_i$ of staying in state $s_i$ the expected return time is given by
\begin{equation}
	m_i=\frac{\sum_{s_j\in S}\mu_j\cdot E(s_j)^{-1}}{\mu_i}.
	\label{eq:expRetTime}
\end{equation}
Gathering those results yields:
\begin{eqnarray*}
	LRA^D(s,G) & = & \mathbb{E}_{s,D}\left(\lim_{t\to\infty}\frac{1}{t}\int_0^t\textbf{1}(\mathcal{X}_u)du\right)= \mathbb{E}_{s,D}\left(\lim_{t\to\infty}\frac{1}{t}\int_0^t\sum_{s_i\in G}\textbf{1}_{\{s_i\}}(\mathcal{X}_u)du\right)\\
	& = & \sum_{s_i \in G} \mathbb{E}_{s,D}\left(\lim_{t\to\infty}\frac{1}{t}\int_0^t\textbf{1}_{\{s_i\}}(\mathcal{X}_u)du\right) \overset{\eqref{eq:fractionTime}}{=} \sum_{s_i \in G} \frac{1}{m_i\cdot E(s_i)}\\
	& \overset{\eqref{eq:expRetTime}}{=} & \sum_{s_i \in G} \frac{\mu_i}{\sum_{s_j\in S}\mu_j\cdot E(s_j)^{-1}} \cdot \frac{1}{E(s_i)} = \frac{\sum_{s_i\in G}\mu_i\cdot E(s_i)^{-1}}{\sum_{s_j\in S}\mu_j\cdot E(s_j)^{-1}} \\
	& = & \frac{\sum_{s_i\in S}\textbf{1}_{G}(s_i)\cdot \mu_i E(s_i)^{-1})}{\sum_{s_j\in S}\mu_j\cdot E(s_j)^{-1}} = \frac{\sum_{s_i\in S}\mu_i\cdot(\textbf{1}_{G}(s_i)\cdot E(s_i)^{-1})}{\sum_{s_j\in S}\mu_j\cdot E(s_j)^{-1}}\\
	& = &  \frac{\sum_{s_i\in S}\mu_i\cdot c_1(s_i,D(s_i))}{\sum_{s_j\in S}\mu_j\cdot c_2(s_j,D(s_j))} \overset{\text{\cite{DBLP:conf/lics/Alfaro98}}}{=} \mathbb{E}_{s,D}(\mathcal{R})
\end{eqnarray*}
Thus, by definition there exists a one to one correspondence between the scheduler $D$ of $\mI$ and its corresponding MDP $\mbox{\sf mdp}(\mI)$. With the results from above this yields that $\LRA^{min}(s,G)=\inf_D\LRA^D(s,G)$ in MA $\mI$ equals $R^{min}(s)=\inf_D\mathbb{E}_{s,D}(\mathcal{R})$ in MDP $\mbox{\sf mdp}(\mI)$.\qed
\section{Proof of Theorem 4}
\begin{theorem}
For MA $\mI = (S, \Act, \it{\,}, \mt{\, }, s_0)$ with MECs $\{ \mI_1, \ldots, \mI_k \}$ 
with state spaces $S_1, \dots, S_k \subseteq S$, and set of goal states $G \subseteq S$:
\begin{align*}
  \LRA^{\min}(s_0,G) & = \inf_{D} \sum_{j=1}^{k} \LRA^{\min}_j(G) \cdot {\Pr}^{D}(s_0 \models \diamondsuit \Box S_j),
\end{align*}
where ${\Pr}^{D}(s_0 \models \diamondsuit \Box S_j)$ is the probability to eventually reach and continuously stay in some state in $S_j$ from $s_0$ under policy $D$ and $\LRA^{\min}_j(G)$ is the LRA of $G \cap S_j$ in unichain MA $\mI_j$.  
\end{theorem}
\textit{Proof}. We give here a sketch proof of Theorem~4. Let $\mI$ be a finite MA with maximal end components $\{\mI_1,\ldots,\mI_k\}$, $G\subseteq S$ a set of goal states, and $\pi\in\Paths(\mI)$ an infinite path in $\mI$. We consider $D$ as a stationery deterministic scheduler. Therefore $\pi$ can be partitioned into an finite and infinite path fragment
\begin{eqnarray*}
\pi_{s_0s}& = &s_0\it{\alpha_0,\mu_0,t_0}s_1 \it{\alpha_1,\mu_1,t_1} \ldots \it{\alpha_n,\mu_n,t_n} s, \text{ and}\\
\pi_{s}^{\omega}& = &s\it{\alpha_s,\mu_s,t_s} \ldots \it{\alpha_i,\mu_i,t_i} s \ldots
\end{eqnarray*}
where $\pi_{s_0s}$ is the path starting in initial state $s_0$ and ends in $s\in\mI_i$. Further, all states on path $\pi^{\omega}_s$ belong to maximal end component $\mI_i$. Note, that a state on path $\pi_{s_0s}$ can be part of another maximal end component $M_j$ as in Example~\ref{ex:append}. Hence, it is not sufficient to only check if eventually a MEC is reached, as done in the corresponding theorem for IMCs in \cite{DBLP:conf/nfm/GuckHKN12}. Thus, the minimal LRA will be obtained when the LRA in each MEC $\mI_i$ is minimal and the combined LRA of all MECs is minimal according to their persistence under scheduler $D$. \qed
\section{Example of Definition 4}\label{sec:exampleLRA}
\setcounter{definition}{3}
\begin{definition}[SSP for long run average]
The SSP of MA $\mathcal{M}$ for the LRA in $G \subseteq S$ is
$\mbox{\sf ssp}_{lra}(\mathcal{M}) = 
\left( S \setminus \smash{\bigcup_{i=1}^k} S_i \cup U \cup Q, \Act \cup \{ \bot \}, \bfP', s_0, Q, c, g \right)$, where
$g(q_i) = \LRA^{\min}_i(G)$ for $q_i \in Q$ and $c(s,\sigma) = 0$ for all $s$ and $\sigma\in \Act\cup\{\bot\}$.
$\bfP'$ is defined as follows. Let $S' =  S \setminus \smash{\bigcup_{i=1}^k} S_i$.  $\bfP'$ equals $\bfP$ for all $s,s' \in S'$.
For the new states $u_j$:
{\small
  \begin{align*}
  \bfP'(u_j, \tau, s') & =  \bfP(S_j, \tau, s') \quad \text{\!if } s' \in S' \setminus S_j 
  &
  \mbox{\!\!and\!\! } \quad \bfP'(u_i, \tau, u_j) & = \bfP(S_i, \tau, S_j) \quad \text{\!for } i \neq j.
\end{align*}}%
Finally, we have: $\bfP'(q_j,\bot,q_j) = 1 = \bfP'(u_j, \bot,q_j)$ and $\bfP'(s, \sigma, u_j) = \bfP(s, \sigma, S_j)$.
\end{definition}
\begin{figure}[t!]
\centering
\subfigure[\scriptsize{Example Markov automata.}]{\label{fig:example_mec}
\begin{tikzpicture}[scale=0.9, every node/.style={transform shape}]
	\node[state, initial] (s0) {$s_0$};
	\node[state] (s3) [right of=s0,node distance=2cm] {$s_1$};
	\node[state] (s5) [below of=s3,node distance=2cm] {$s_3$};
	\node[state, accepting] (s6) [right of=s3,node distance=2cm] {$s_2$};
	\node[state] (s8) [left of=s5,node distance=2cm] {$s_5$};
	\node[state] (s9) [right of=s5,node distance=2cm] {$s_4$};
	
	\node[rectangle,draw,green,fit=(s3)(s5)(s6)(s9)] (mec1) {};
	\node[rectangle,draw,red,fit=(s8)] (mec2) {};
	
	\path[->] (s0) edge[double,thin] node[sloped,above] {$2$} (s3);
	\path[->] (s3) edge[thin] node[left] {$0.6$} (s5);
	\path[->] (s3) edge[thin] node[sloped,above] {$0.4$} (s6);
	
	\path[->]
		(s3) edge[thin] node[inner sep=0mm,pos=0.2] (a1) {} (s5)
		(s3) edge[thin] node[inner sep=0mm,pos=0.2] (b1) {} (s6);	
	\path[-,shorten <=-.4pt,shorten >=-.4pt] (a1) edge [thin,bend right]  (b1) node[right,yshift=.2cm,xshift=.4cm] {$\alpha$} ;
	
	\path[->] (s5) edge[thin] node[sloped,above] {$\alpha,1$} (s8);
	\path[->] (s5) edge[thin] node[sloped, above] {$\beta,1$} (s9);
	\path[->] (s8) edge[thin,double,loop left] node[left] {$1$} (s8);
	\path[->] (s9) edge[thin,double] node[right] {$3$} (s6);
	\path[->] (s6) edge[thin,double, bend right=60] node[sloped,above] {$1$} (s3);
\end{tikzpicture}
}
\hspace{.3cm}
\subfigure[\scriptsize{Induced SSP for MA in Figure \ref{fig:example_mec}.}]{\label{fig:example_ssp}
\begin{tikzpicture}[scale=0.9, every node/.style={transform shape}]
	\node[state, initial] (s0) {$s_0$};
	\node[state,fill=green] (u2) [right of=s0,node distance=2cm] {$u_1$};
	\node[state,fill=green] (q2) [right of=u2,node distance=2cm] {$q_1$};
	\node[state,fill=red] (u1) [below of=u2,node distance=2cm] {$u_2$};
	\node[state,fill=red] (q1) [right of=u1,node distance=2cm] {$q_2$};
	
	\path[->] (s0) edge[thin] node[sloped,above] {$\bot,1$} (u2);
	\path[->] (u2) edge[thin] node[sloped,above] {$\bot,1$} (q2);
	\path[->] (u2) edge[thin] node[left] {$\alpha,1$} (u1);
	\path[->] (u1) edge[thin] node[sloped,above] {$\bot,1$} (q1);
	\path[->] (q2) edge[thin,loop right] node[right] {$\bot,1$} (q2);
	\path[->] (q1) edge[thin,loop right] node[right] {$\bot,1$} (q1);
\end{tikzpicture}
}
\caption {Example for Definition 4.}
\end{figure}
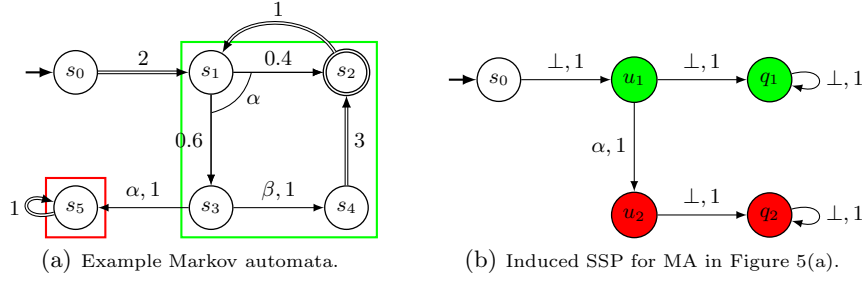
\begin{example}
 Consider the MA $\mI$ from Figure \ref{fig:example_mec} with MECs $\mI_1$ with $S_1=\{s_1,s_2,s_3,s_4\}$ and $\mI_2$ with $S_2=\{s_5\}$.  We construct the corresponding $\mbox{\sf ssp}_{lra}(\mI)$ due to Definition~4. Let $S_{\mbox{\sf ssp}} = S \setminus \smash{\bigcup_{i=1}^k} S_i \cup U \cup Q$, where  $\smash{\bigcup_{i=1}^k} S_i = S_1 \cup S_2 = \{s_1,s_2,s_3,s_4,s_5\}$. Further, we have to MECs and therefore fresh states $U=\{u_1,u_2\}$ and $Q=\{q_1,q_2\}$. Hence, $S_{\mbox{\sf ssp}} = \{s_0,u_1,u_2,q_1,q_2\}$. (1) Consider $s,s'\in S'$. Since, $S'=\{s_0\}$ and there exists no transition from $s_0$ to $s_0$ we can omit this rule. (2) Consider outgoing transitions from MECs. For $\mI_1$ there exists a transition from $s_3 \it{\alpha,1} s_5$ in the underlying MA, where $s_3\in S_1$ and $s_5\not\in S_1$ but $s_5 \in S_2$. For the corresponding new state $u_1$ it follows $\bfP'(u_1,\alpha,u_2)=\bfP(S_1,\alpha,S_2)=1$ where $\bfP(S_i,\sigma,S_j)=\sum_{s \in S_i}\sum_{s' \in S_j} \bfP(s,\sigma,s')$. (3) Consider all states $U$ and $Q$ and add new transitions with $\bfP(u_i,\bot,q_i) = \bfP(q_i,\bot,q_i)=1$ for $i=1,2$. Finally, consider all states $s \in S_{\mbox{\sf ssp}} \cap S$ with a transition into a MEC. Hence, $\bfP'(s_0,\bot,u_1)=\bfP(s_0,\bot,s_1)=1$. The resulting transition system of $\mbox{\sf ssp}_{lra}(\mI)$ is depicted in Figure~\ref{fig:example_ssp}.
 \label{ex:append}
\end{example}
\section{Proof of Theorem 5}
\begin{theorem} 
For MA $\mathcal{M}$, $\LRA^{\min}(s,G)$ equals $cR^{\min}(s,\diamondsuit U)$ in SSP $\mbox{\sf ssp}_{lra}(\mathcal{M})$.
\end{theorem}
\textit{Proof}. We show that the reduction of the induced SSP is correct.
\begin{eqnarray*}
	cR^{min}(s,\F Q) & = & \inf_{D}\mathbb{E}_{s,D}\{ g(X_{T_Q}) \} = \inf_D \sum_{i=1}^{k} g(X_{T_{q_i}})\cdot Pr^D(s\models \F q_i)\\
	& = & \inf_D \sum_{i=1}^{k} \LRA_{i}^{min}(G)\cdot Pr^D(s\models\F q_i)\\
	&\overset{(*)}{=}  & \inf_D \sum_{i=1}^{k} \LRA_{i}^{min}(G)\cdot Pr^D(s\models\F\G S_i)\\
	& = & \LRA^{min}(s,G).
\end{eqnarray*}
Observe that in step $(*)$ we use the transformation from Definition 4 in reverse. Hence, if $Pr^D(s\models\F q_i)>0$, we eventually reach the maximal end component $M_i$ and always stay in it. Otherwise $Pr^D(s\models\F q_i)=0$ and scheduler $D$ chooses an action such that we leave $M_i$ or never even visit $M_i$.\qed
\section{Proof of Theorem \ref{thm:tbr}}
\def\hash{{\scriptstyle\#}} We assume the settings of
Theorem~\ref{thm:tbr} to hold: MA \defma is given together with a set
of goal states $G \subseteq S$, time interval $I=[0,b] \in
\mathcal{Q}$, $b \ge 0$. Let $\lambda = \max_{s \in \MaS}E(s)$ be the
largest exit rate of any Markovian state and $\delta > 0$ be chosen
such that $b=k_b\delta$ for some $k_b \in \mathbb{N}$. We recall the
definiton of $\diamondsuit^IG$ as the set of all paths that reach the
goal states in $G$ within interval $I$. We also define a random
variable $\hash_{J}: \paths \rightarrowtail \mathbb{N}$, where $J \in
\mathcal{Q}$ is a time interval. Intuitively $\hash_J$ counts the
number of Markovian jumps inside interval $J$. For example
$\hash_{[0,\delta]}=1$ denotes the set of paths having one Markovian
transition in their first $\delta$ time units. Random vector
$\hash^{J,\Delta}: \paths \rightarrowtail \mathbb{N}^k$ with $J \in
\mathcal{Q}$, $k$ such that $k\delta=\sup J$ and $\Delta \in
\mathbb{Q}_{>0}$ is defined as the vector of Markovian jump counts in
each subinterval (digitisation step) of size $\Delta$. For instance
$\hash^{I,\delta}(\pi)$ with $\pi \in \paths$ is vector
$\left(\hash_{[0,\delta)}, \dotsc,
  \hash_{[(k_b-2)\delta,(k_b-1)\delta)},
  \hash_{[(k_b-1)\delta,b]}\right)^{\text{T}}$.

\begin{lemma}\label{lemma:dma}
  Let $\MAM_{\delta}$ be the dMA induced by \MAM with respect to
  digitisation constant $\delta$. Then for all $s \in S$:
  \begin{equation*}
    p^{\MAM_{\delta}}_{\max}(s, \rchgls{[0,k_b]})=\sup_{D \in \GMS} {\Pr}_{s,D}(\rchgls{I} \mid \left \|\hash^{I,\delta}  \right \|_{\infty} < 2).
  \end{equation*}
\end{lemma}
\begin{proof}
  As we discussed in Section~\ref{section:timed}, paths of
  $\MAM_{\delta}$ are essentially the path from \MAM that carry only
  zero or one Markov transitions in each digitisation step
  $\delta$. For computing reachability in step interval $[0,k_b]$ in
  $\MAM_{\delta}$, it is enough to consider paths in \MAM bearing at
  most one Markovian jumps in each $\delta$ time units. This set of
  paths can be described by $\left \| \hash^{I,\delta} \right \|_{\infty}
  < 2$.\qed
\end{proof}

\begin{lemma}\label{lemma:lb1} For all $s \in S$ and $D \in \GMS$ in
  \MAM, ${\Pr}_{s,D}(\rchgls{I} \mid \hash_{[0,\delta]}<2) \le
  {\Pr}_{s,D}(\rchgls{I})$.
\end{lemma}

\begin{proof}
  We assume $b>0$, since for $b=0$, ${\Pr}_{s,D}(\rchgls{I} \mid
  \hash_{[0,\delta]}<2)={\Pr}_{s,D}(\rchgls{I})$. We have
\begin{flalign}
 {\Pr}_{s,D}&(\rchgls{I}) &\nonumber \\&={} {\Pr}_{s,D}(\rchgls{I} \cap \hash_{[0,\delta]} > 0) + {\Pr}_{s,D}(\rchgls{I} \cap \hash_{[0,\delta]} = 0)& \nonumber\\
 &={} {\Pr}_{s,D}(\rchgls{I} \cap \hash_{[0,\delta]} > 0) + {\Pr}_{s,D}(\rchgls{I} \mid \hash_{[0,\delta]} = 0){\Pr}_{s,D}(\hash_{[0,\delta]} = 0).\hspace{-10pt}& \label{eq:tbr}
\end{flalign}
On the other hand we have
\begin{flalign}
 {\Pr}_{s,D}(\rchgls{I} &\mid \hash_{[0,\delta]}<2)& \nonumber \\
 = {} & {\Pr}_{s,D}(\rchgls{I} \mid \hash_{[0,\delta]} <2, \hash_{[0,\delta]} =1 ){\Pr}_{s,D}( \hash_{[0,\delta]} =1 \mid \hash_{[0,\delta]} <2)& \nonumber \\ &+ 
 {\Pr}_{s,D}(\rchgls{I} \mid \hash_{[0,\delta]} <2, \hash_{[0,\delta]} =0 ){\Pr}_{s,D}( \hash_{[0,\delta]} = 0 \mid \hash_{[0,\delta]} <2).\hspace{-10pt}& \label{eq:tbrapprox}
\end{flalign}
We distinguish between two cases:
\begin{enumerate}
\item $s \in \MS \setminus G$: In this case,~\eqref{eq:tbr} gives
  \begin{flalign}
    {\Pr}_{s,D}(\rchgls{I}) = {} & \int_0^{\delta} E(s) e^{-E(s)t} \sum_{s' \in S} \bfP (s,\bot,s') \Pr_{s',D}(\rchgls{I \ominus t})\dd t \nonumber &\\
    &+ {\Pr}_{s,D}(\rchgls{I \ominus
      \delta})e^{-E(s)\delta}.& \label{eq:tbrms}
  \end{flalign}
  and for~\eqref{eq:tbrapprox} we have
  \begin{flalign}
    {\Pr}_{s,D}(\rchgls{I}\mid \hash_{[0,\delta]}<2) = {} &\int_0^{\delta} E(s) e^{-E(s)t} \sum_{s' \in S} \bfP (s,\bot,s') \Pr_{s',D}(\rchgls{I \ominus\delta})\dd t \nonumber &\\
    &+ {\Pr}_{s,D}(\rchgls{I \ominus
      \delta})e^{-E(s)\delta}.&\label{eq:tbrapproxms}
  \end{flalign}
  Since ${\Pr}_{s,D}(\rchgls{I \ominus t})$ is monotonically
  decreasing with respect to $t$, we have ${\Pr}_{s,D}(\rchgls{I
    \ominus \delta}) \le {\Pr}_{s,D}(\rchgls{I \ominus t}), \; t \le
  \delta$. Putting this in~\eqref{eq:tbrms} and~\eqref{eq:tbrapproxms}
  leads to
  \begin{equation*}
    {\Pr}_{s,D}(\rchgls{I} \mid \hash_{[0,\delta]}<2) \le {\Pr}_{s,D}(\rchgls{I})
  \end{equation*}
\item $s \in \IS \setminus G$: From the law of total probability, we
  split time bounded reachability into two parts. First we compute the
  probability to reach the set of Markovian states from $s$ by only
  taking interactive transitions in zero time, and then we quantify
  the probability to reach the set of goal states $G$ from Markovian
  states inside interval $I$. Therefore:
\begin{flalign*}
{\Pr}_{s,D}(\rchgls{I}) = {} & \sum_{s' \in \MS} {\Pr}_{s,D}(\diamondsuit^{[0,0]}\{s'\}) {\Pr}_{s',D}(\rchgls{I}) \\
        \overset{(*)}{\ge}{} & \sum_{s' \in \MS} {\Pr}_{s,D}(\diamondsuit^{[0,0]}\{s'\}) {\Pr}_{s',D}(\rchgls{I}\mid \hash_{[0,\delta]} < 2 ) \\
                        = {} & {\Pr}_{s,D}(\rchgls{I}\mid \hash_{[0,\delta]} < 2 )
\end{flalign*}
where $(*)$ follows from case (i) above.\qed
\end{enumerate}
\end{proof}

\begin{lemma}\label{lemma:lb2} For all $s \in S \setminus G$ and $D \in \GMS$ in \MAM,
  ${\Pr}_{s,D}(\rchgls{I} \mid \left \|\hash^{I,\delta}  \right \|_{\infty} < 2) \le {\Pr}_{s,D}(\rchgls{I} \mid \hash_{[0,\delta]}<2)$.
\end{lemma}
\begin{proof}
  The lemma holds for $b=0$, since in this case,
  ${\Pr}_{s,D}(\rchgls{I} \mid \left \|\hash^{I,\delta} \right
  \|_{\infty} < 2) = {\Pr}_{s,D}(\rchgls{I} \mid
  \hash_{[0,\delta]}<2)$. For $b>0$, we decompose
  ${\Pr}_{s,D}(\rchgls{I} \mid \left \|\hash^{I,\delta} \right
  \|_{\infty} < 2)$ as \eqref{eq:tbrapprox} to:
  \begin{flalign} {\Pr}_{s,D}(\diamondsuit^I&G\mid \left
      \|\hash^{I,\delta} \right \|_{\infty} < 2))& \nonumber \\ = {} &
    {\Pr}_{s,D}(\rchgls{I} \mid \left \|\hash^{I,\delta} \right
    \|_{\infty} < 2, \hash_{[0,\delta]} =1 ){\Pr}_{s,D}(
    \hash_{[0,\delta]} =1 \mid \left \|\hash^{I,\delta} \right
    \|_{\infty} < 2)& \nonumber \\&\hspace{-4pt}+
    {\Pr}_{s,D}(\rchgls{I} \mid \left \|\hash^{I,\delta} \right
    \|_{\infty} < 2, \hash_{[0,\delta]} =0 ){\Pr}_{s,D}(
    \hash_{[0,\delta]} = 0 \mid \left \|\hash^{I,\delta} \right
    \|_{\infty} < 2).& \label{eq:tbrdig}
  \end{flalign}
  We proof the lemma by induction over $k_b$.
  \begin{itemize}
  \item $k_b=1$: This case holds because interval $I=[0,\delta]$ contains one digitisation step and then ${\Pr}_{s,D}(\rchgls{I} \mid \left \|\hash^{I,\delta}  \right \|_{\infty} < 2) = {\Pr}_{s,D}(\rchgls{I} \mid \hash_{[0,\delta]}<2)$.
  \item $k_b - 1 \leadsto k_b$: Let $I$ be $[0,b]$ and assume the
    lemma holds for interval $[0,(k_b -1)\delta]$
    (i.e. $I\ominus\delta$):
  \begin{equation}
    {\Pr}_{s,D}(\rchgls{I\ominus\delta} \mid \left \|\hash^{I\ominus\delta,\delta}  \right \|_{\infty} < 2) \le {\Pr}_{s,D}(\rchgls{I\ominus\delta} \mid \hash_{[0,\delta]}<2).\label{eq:ih}
  \end{equation}
  In order to show that the lemma holds for $I$, we distinguish
  between two cases:
  \begin{enumerate}
  \item $s \in S \setminus \MS$: From~\eqref{eq:tbrapproxms} we have:
    \begin{flalign}
      {\Pr}_{s,D}(\rchgls{I}\mid \hash_{[0,\delta]}<2) = {} &\sum_{s' \in S} \bfP (s,\bot,s') \Pr_{s',D}(\rchgls{I \ominus\delta})(1-e^{-E(s)\delta}) \nonumber &\\
      &+ {\Pr}_{s,D}(\rchgls{I \ominus
        \delta})e^{-E(s)\delta}.&\label{eq:tbrapproxmssim}
    \end{flalign}
    Similarly from~\eqref{eq:tbrdig} we have:
    \begin{flalign*}
      {\Pr}_{s,D}(\rchgls{I}\mid &\left \| \hash^{I,\delta}  \right \|_{\infty} < 2)&  \\= {} &\sum_{s' \in S} \bfP (s,\bot,s') \Pr_{s',D}(\rchgls{I \ominus\delta} \mid \left \|\hash^{I\ominus\delta,\delta}  \right \|_{\infty} < 2)(1-e^{-E(s)\delta}) &\\
      &+ {\Pr}_{s,D}(\rchgls{I \ominus
        \delta}\mid \left \|\hash^{I\ominus\delta,\delta}  \right \|_{\infty} < 2)e^{-E(s)\delta}&\\
      \overset{\eqref{eq:ih}}{\le} {} & \sum_{s' \in S} \bfP (s,\bot,s') \Pr_{s',D}(\rchgls{I \ominus\delta})(1-e^{-E(s)\delta}) &\\
      &+ {\Pr}_{s,D}(\rchgls{I \ominus
        \delta})e^{-E(s)\delta}&\\
      \overset{\eqref{eq:tbrapproxmssim}}{=} {} &
      {\Pr}_{s,D}(\rchgls{I}\mid \hash_{[0,\delta]}<2)&
    \end{flalign*}
  \item $s \in S \setminus \IS$: This case utilises the previously
    discussed idea of splitting paths using the law of total
    proabilities into two parts. The first part contains the set of
    paths that reach Markovian states from $s$ in zero time using
    interactive transitions, while the second includes paths reaching
    $G$ from Markovian states. Hence:
    \begin{flalign*}
      {\Pr}_{s,D}(\rchgls{I}\mid \left \|\hash^{I,\delta}  \right \|_{\infty} &< 2)&\\ = {} & \sum_{s' \in \MS} {\Pr}_{s,D}(\diamondsuit^{[0,0]}\{s'\}) {\Pr}_{s',D}(\rchgls{I}\mid \left \|\hash^{I,\delta}  \right \|_{\infty} < 2)& \\
      \overset{(*)}{\le}{} & \sum_{s' \in \MS} {\Pr}_{s,D}(\diamondsuit^{[0,0]}\{s'\}) {\Pr}_{s',D}(\rchgls{I}\mid \hash_{[0,\delta]} < 2 )& \\
      = {} & {\Pr}_{s,D}(\rchgls{I}\mid \hash_{[0,\delta]} < 2 )
    \end{flalign*}
    where $(*)$ follows from case (i) above.\qed
  \end{enumerate}
  \end{itemize}
\end{proof}

\begin{lemma}\label{lemma:lb} For all $s \in S \setminus G$:
  $p^{\MAM_{\delta}}_{\max}(s,\rchgls{[0,k_b]}) \le p^{\MAM}_{\max}(s,\rchgls{I})$.
\end{lemma}
\begin{proof}
\begin{align*}
 p^{\MAM_{\delta}}_{\max}(s, \rchgls{[0,k_b]})  \overset{Lem.~\ref{lemma:dma}}{=} {} & \sup_{D \in \GMS} {\Pr}_{s,D}(\rchgls{I} \mid \left \|\hash^{I,\delta}  \right \|_{\infty} < 2)\\
\overset{Lem.~\ref{lemma:lb1}}{\le} {} & \sup_{D \in \GMS} {\Pr}_{s,D}(\rchgls{I} \mid  \hash_{[0,\delta]} < 2) \\
\overset{Lem.~\ref{lemma:lb2}}{\le} {} & \sup_{D \in \GMS} {\Pr}_{s,D}(\rchgls{I}) =  p^{\MAM}_{\max}(s,\rchgls{I})
\end{align*}\qed
\end{proof}

\begin{lemma}\label{lemma:ub} For all $s \in S \setminus G$:
\begin{equation*}
  p^{\MAM}_{\max}(s,\rchgls{I}) \le p^{\MAM_{\delta}}_{\max}(s,\rchgls{[0,k_b]}) + 1 - e^{-\lambda b}(1 + \lambda\delta)^{k_b}.
\end{equation*}
\end{lemma}

\begin{proof}
  \begin{flalign*}
    p^{\MAM}_{\max}(s,&\rchgls{I})& \\ = {}& \sup_{D \in \GMS} {\Pr}_{s,D}(\rchgls{I})&\\
                                 = {}& \sup_{D \in \GMS} \Big({\Pr}_{s,D}(\rchgls{I} \cap \left \| \hash^{I,\delta} \right \|_{\infty} < 2 ) + {\Pr}_{s,D}(\rchgls{I} \cap \left \| \hash^{I,\delta} \right \|_{\infty} \ge 2)\Big)& \\
                                 \le {}& \sup_{D \in \GMS} {\Pr}_{s,D}(\rchgls{I} \cap \left \| \hash^{I,\delta} \right \|_{\infty} < 2 ) + \sup_{D \in \GMS}{\Pr}_{s,D}(\rchgls{I} \cap \left \| \hash^{I,\delta} \right \|_{\infty} \ge 2)& \\
                                 \le {}& \sup_{D \in \GMS} {\Pr}_{s,D}(\rchgls{I} \mid \left \| \hash^{I,\delta} \right \|_{\infty} < 2 ) + \sup_{D \in \GMS}{\Pr}_{s,D}(\rchgls{I} \cap \left \| \hash^{I,\delta} \right \|_{\infty} \ge 2)& \\
                                 = {}& p^{\MAM_{\delta}}_{\max}(s,\rchgls{[0,k_b]}) + \sup_{D \in \GMS}{\Pr}_{s,D}(\rchgls{I} \cap \left \| \hash^{I,\delta} \right \|_{\infty} \ge 2)& \\
                                 \le {}& p^{\MAM_{\delta}}_{\max}(s,\rchgls{[0,k_b]}) + \sup_{D \in \GMS}{\Pr}_{s,D}(\left \| \hash^{I,\delta} \right \|_{\infty} \ge 2)& \\
  \end{flalign*}
  It remains to find an upper bound for $\sup_{D \in
    \GMS}{\Pr}_{s,D}(\left \| \hash^{I,\delta} \right \|_{\infty} \ge
  2$ which is the maximum probability to have more than one Markovian
  jump in at least one time step among $k_b$ time step(s) of length
  $\delta$.  Due to independence of the number of Markovian jumps in
  digitisation steps, this probability can be upper bounded by $k_b$
  independent Poisson processes, all parametrised with the maximum
  exit rate exhibited in \MAM. In each Poisson process the probability
  of at most one Markovian jump in one digitisation step is $e^{-
    \lambda \delta}(1 + \lambda \delta)$, therefore the probability of
  a violation of this assumption in at least one digitisation step is
  $1 - e^{- \lambda b}\big(1+ \lambda \delta\big)^{k_b}$. Hence
  \begin{align*}
    p^{\MAM}_{\max}(s,\rchgls{I}) \le {}& p^{\MAM_{\delta}}_{\max}(s,\rchgls{[0,k_b]}) + \sup_{D \in \GMS}{\Pr}_{s,D}(\left \| \hash^{I,\delta} \right \|_{\infty} \ge 2)& \\
                                  \le {}& p^{\MAM_{\delta}}_{\max}(s,\rchgls{[0,k_b]}) + 1 - e^{- \lambda b}\big(1+ \lambda \delta\big)^{k_b}
  \end{align*}\qed
\end{proof}

\begin{theorem}
    Given MA $\mathcal{M}=(S,\Act, \it{}, \mt{}, s_0)$, $G \subseteq S$, interval 
    $I=[0,b] \in \mathcal{Q}$ with $b \ge 0$ and $\lambda = \max_{s \in \scriptsize\MS}E(s)$. 
    Let $\delta > 0$ be such that $b=k_b\delta$ for some $k_b \in \mathbb{N}$. 
    Then, for all $s \in S$ it holds that
    \begin{equation*}
      p^{\mathcal{M}_{\delta}}_{\max}(s, \diamondsuit^{[0,k_b]} \, G) 
      \ \leq \ 
      p^{\mathcal{M}}_{\max}(s, \diamondsuit^{[0,b]} \, G) 
      \ \leq \ 
      p^{\mathcal{M}_{\delta}}_{\max}(s, \diamondsuit^{[0,k_b]} \, G) + 1 - e^{- \lambda b}\big(1+ \lambda \delta\big)^{k_b}.
    \end{equation*}
\end{theorem}
\begin{proof}
  For $s\in G$ we have that $p^{\mathcal{M}_{\delta}}_{\max}(s,
  \diamondsuit^{[0,k_b]} \, G) = p^{\mathcal{M}}_{\max}(s,
  \diamondsuit^{[0,b]} \, G)=1$. For $s \in S \setminus G$ it follows
  from Lemma~\ref{lemma:lb} and~\ref{lemma:ub}.
\end{proof}
\fi
\end{document}